\documentclass[journal]{IEEEtran}
\usepackage{cite}
\ifCLASSINFOpdf
\else
\fi

\usepackage{array}
\usepackage{algorithm}
\usepackage{algpseudocode}
\usepackage{amsmath}
\usepackage{amssymb}
\usepackage{booktabs}
\usepackage{relsize}
\usepackage{mathtools}
\usepackage{siunitx}
\usepackage{makecell}
\usepackage{tabularx}
\usepackage{tikz}
\usepackage{subcaption}

\newcommand{\etal}{\emph{et al.}\@}

\usepackage{amsthm}
\newtheorem{lemma}{Lemma}
\newtheorem{theorem}{Theorem}
\newtheorem{remark}{Remark}
\title{BC-DIR: Bandit-Controlled Deadline-Aware Incremental Redundancy for QUIC in V2X Networks}

\author{Weihao~Liu,~\IEEEmembership{Student Member,~IEEE,}
        He~Sun,~\IEEEmembership{Member,~IEEE,}
        and~Chen-Khong~Tham,~\IEEEmembership{Senior Member,~IEEE}
\thanks{Weihao Liu and Chen-Khong Tham are with the Department of Electrical and Computer Engineering, National University of Singapore, Singapore 117583 (e-mail: liuweihao@u.nus.edu; eletck@nus.edu.sg).}%
\thanks{He Sun is with the School of Electronic and Information Engineering, Beihang University, Beijing 100191, China (e-mail: sunhe1710@buaa.edu.cn).}%
}

\begin{document}
\maketitle

\begin{abstract}
Vehicle-to-Everything (V2X) communications require timely and reliable message delivery under highly dynamic wireless conditions. Existing approaches that integrate forward error correction (FEC) into QUIC rely mainly on proactive redundancy and fall back to retransmission once losses exceed the correction capability of the configured code, leading to inefficient recovery under burst loss and unnecessary overhead when network conditions are favorable. This paper presents a Bandit-Controlled Deadline-Aware Incremental Redundancy (BC-DIR) framework for QUIC-based V2X transport. BC-DIR combines rateless coding with a soft decoding deadline and a redundancy margin, enabling repair to be triggered within the available delivery budget while injecting additional repair symbols beyond the immediate deficit to improve recovery under burst loss. A contextual bandit controller further adapts the redundancy configuration online according to end-to-end feedback. We also develop a deadline-constrained reliability analysis under burst loss, showing the advantage of the proposed repair mechanism over conventional retransmission and the existence of an optimal redundancy margin. Monte Carlo simulations validate the analytical results. BC-DIR is implemented in a QUIC-based transport stack and evaluated in Veins/OMNeT++ under both congested urban V2X scenarios and stable network conditions. Experimental results show that, across different traffic congestion levels, BC-DIR improves completion ratio by 10\%--40\% over benchmark schemes in congested V2X scenarios, while under favorable network conditions it can even reduce overhead by about 1\% compared with native QUIC.
\end{abstract}


%
\IEEEpeerreviewmaketitle

\section{Introduction}
\label{sec:intro}

V2X communication is a key enabler for intelligent transportation systems and autonomous driving~\cite{hejazi2022survey}. By allowing vehicles to exchange perception data, map updates, and cooperative control information, V2X extends sensing beyond line-of-sight (LoS) and improves both safety and traffic efficiency. However, reliable and timely data delivery in V2X networks remains challenging~\cite{gyawali2020challenges}. In dense urban environments, frequent topology changes, wireless interference, and channel contention cause bursty packet losses, while many V2X applications are highly delay-sensitive~\cite{wang2019survey,eckermann2019performance}. These characteristics make transport design in V2X fundamentally different from that in conventional wired networks.

Conventional transport protocols are not well suited to this setting. Transmission Control Protocol (TCP) provides reliability through retransmission and congestion control, but its head-of-line blocking and retransmission delay can significantly degrade timeliness under lossy wireless conditions~\cite{kyratzis2022quic}. User Datagram Protocol (UDP) avoids retransmission delay but provides no reliability guarantee. Quick UDP Internet Connections (QUIC) has recently emerged as a promising alternative, combining congestion control, multiplexing, and encryption on top of UDP while reducing connection establishment latency and mitigating application-layer head-of-line blocking~\cite{amponis2024channel,suleman2026simulated}. These properties make QUIC attractive for delay-sensitive wireless scenarios. Nevertheless, QUIC still ensures reliability through loss detection and retransmission~\cite{vu2020latency}, which may consume additional round-trip time (RTT) and cause deadline misses under burst loss.

FEC can alleviate this problem by enabling local recovery without waiting for retransmissions~\cite{mathis1996forward,polardeletion}. However, FEC inevitably introduces extra redundancy overhead. When channel conditions are good, it leads to direct bandwidth waste, which limits the practical use of proactive FEC. As a result, even in the presence of losses, retransmission-based recovery is often preferred in practice because it avoids persistent redundancy cost~\cite{9756315,sun2024fast}.

Existing QUIC-FEC approaches improve reliability by proactively transmitting repair symbols and, in some cases, adapting the coding rate to changing network conditions~\cite{garrido2019rquic,zhang2025qrst,michel2019quicfec,sidhu2026tarot}. However, their protection remains fundamentally proactive: even when channel conditions are favorable, FEC symbols still need to be transmitted, which introduces persistent bandwidth overhead. Moreover, once packet losses exceed the correction capability of the configured FEC block, recovery falls back to QUIC's native automatic repeat request (ARQ) mechanism, incurring additional retransmission cycles and delay inflation. These designs implicitly treat loss recovery as a best-effort process, without explicitly controlling when repair should be triggered and how aggressively redundancy should be injected under a deadline constraint. As a result, they fail to fully utilize the limited repair opportunities available before the deadline, especially under burst losses.

To address these limitations, we propose BC-DIR for V2X transport. BC-DIR targets deadline-constrained block recovery by jointly coordinating the initial FEC configuration and the subsequent redundant repair process. Unlike conventional FEC schemes that rely on either initial redundancy or reactive retransmission, DIR introduces a soft decoding deadline to explicitly regulate repair timing, ensuring that recovery actions are aligned with the remaining delivery budget. In particular, BC-DIR introduces a redundancy margin in each repair round, which proactively injects additional repair symbols beyond the immediate deficit to improve the probability of successful decoding under burst losses. This additional redundancy increases the probability that the current block can be recovered within the available deadline budget. The overall redundancy configuration is further adjusted online according to observed network conditions. Under severe burst loss, BC-DIR exploits the joint effect of initial FEC protection and the DIR process to improve deadline-constrained recovery. When channel conditions are favorable, it suppresses unnecessary proactive redundancy and relies on lightweight redundant repair to handle residual losses. In this way, BC-DIR achieves an effective balance between completion probability and transmission overhead under dynamic V2X conditions. The main contributions of this work are summarized as follows:

\begin{itemize}
\item We propose a deadline-aware transport mechanism for V2X networks that introduces a soft decoding deadline to regulate when repair is triggered, so that block-level recovery remains aligned with the end-to-end delivery deadline. The mechanism further incorporates a redundancy margin, which intentionally injects additional repair symbols beyond the immediate deficit in each repair round to improve recovery probability under burst loss. A bandit controller is used to adapt the redundancy configuration online from end-to-end feedback, allowing the proposed design to balance completion ratio and redundancy overhead under dynamic V2X conditions.

\item We develop a theoretical framework for deadline-constrained reliability under burst loss and analytically characterize the advantage of DIR over conventional retransmission. In particular, we derive the one-round repair success model and establish a sufficient condition under which DIR achieves higher completion probability than ARQ, and show the existence of an optimal redundancy margin under a hard deadline. These results provide the analytical basis for online redundancy adaptation.

\item We implement BC-DIR in a QUIC-based transport stack and integrate it into a vehicular network simulation framework. Extensive evaluations under multiple traffic densities and network conditions show that BC-DIR consistently outperforms benchmark methods in deadline-constrained completion probability. The performance gain becomes more significant as vehicle density increases and the wireless channel becomes more bursty and congested, while under favorable network conditions BC-DIR still maintains competitive bandwidth efficiency.
\end{itemize}

The remainder of this paper is organized as follows. Section~\ref{sec:related} reviews related work on QUIC-based transport and FEC mechanisms in wireless networks. Section~\ref{sec:system} presents the system model. Section~\ref{sec:bandit_ir_quic} details the design of BC-DIR. Section~\ref{sec:evaluation} evaluates its performance under various vehicular scenarios. Section~\ref{sec:conclusion} concludes the paper.

\section{Related Work}
\label{sec:related}

QUIC standardized by the IETF~\cite{rfc9000} detects losses using time-based and packet-threshold mechanisms, and recovers them through retransmission. While this design ensures reliability, each loss event may incur at least one additional RTT before recovery, which can significantly inflate latency in lossy wireless environments.

To reduce retransmission delay, both the IETF community~\cite{swett-nwcrg-coding-for-quic-04} and industrial organizations~\cite{deconinck2019pluginizing,garrido2019rquic,michel2019quicfec} have explored integrating FEC into QUIC. By proactively transmitting repair symbols together with source data, these approaches enable local loss recovery without waiting for explicit retransmission feedback. Early designs considered simple XOR and Reed-Solomon (RS) based coding~\cite{zheng-quic-fec-extension-01}. Coninck~\etal~\cite{deconinck2019pluginizing} proposed Pluginized QUIC (PQUIC), a plugin-based QUIC framework that supports extensible protocol behaviors and implements Google’s earlier XOR-based FEC design. Michel~\etal~\cite{michel2019quicfec} further incorporated multiple coding schemes, including random linear coding (RLC) and RS, into this framework. Later, Michel~\etal~\cite{michel2023flec} introduced Flexible Erasure Correction (FlEC), which apply adaptive FEC-rate adjustment, and QUIC Integrated Reliability Layers (QUIRL)~\cite{michel2024quirl} further improved deployability in practical environments. However, these studies mainly focus on independent and identically distributed (i.i.d) loss or large-file transfer scenarios. Their protection is typically configured proactively, and robust performance under burst-loss-dominated, deadline-sensitive small-block transmissions remains insufficiently explored.

For V2X applications, whether each block can be delivered before its deadline is often more important than maximizing long-term throughput. Zhang~\etal~\cite{zhang2025qrst} proposed QRST, which combines adaptive FEC with deadline-driven block scheduling to improve the number of blocks delivered on time. Yang~\etal~\cite{yang2024qos} proposed QC-MAB, which applies contextual bandit control to jointly adjust transport-layer parameters. These works demonstrate the value of deadline-aware adaptation, but QRST introduces relatively high redundancy overhead, while QC-MAB primarily targets multipath video streaming and quality-of-service (QoS) provisioning rather than single-path V2X communication with burst losses and stringent per-message deadlines.

The ARQ-based recovery is also known to be inefficient over lossy links, since a lost packet may require multiple feedback rounds before successful delivery. To address this limitation, several studies have explored coded retransmission or incremental-redundancy style recovery at the transport layer. Cloud~\etal~\cite{cloud2015coded} extended selective-repeat ARQ with coded retransmissions to reduce in-order delivery delay. Meng~\etal~\cite{meng2024hairpin} further considered deadline miss rate in interactive video streaming and protected retransmitted symbols with additional FEC. However, because its design is based on a fixed-rate FEC scheme, lost packets must be regrouped and re-encoded in subsequent rounds, which increases system complexity and computational overhead. More generally, these studies show the benefit of coded repair over pure ARQ, but do not explicitly address deadline-constrained block recovery under burst-loss-dominated vehicular channels.

Furthermore, several studies have considered using polar codes at the transport layer to enhance reliability~\cite{pereira2022polar, liu2024polarquic, liu2026lstmpolar}. However, the code length of conventional polar codes is limited to powers of two. Thus, rate matching is required for encoding polar codes with arbitrary code lengths via shortening, puncturing or bit padding, resulting in error-correction performance degradation or rate loss. Sidhu~\etal~\cite{sidhu2026tarot} introduced TAROT, which applies RaptorQ to adaptive video streaming and dynamically adjusts coding parameters to reduce FEC overhead. While TAROT demonstrates the potential of rateless coding, its control strategy is designed for long video segments and does not fully exploit the advantages of rateless repair in short-block, deadline-constrained settings.

Overall, existing QUIC-FEC designs mainly rely on proactive redundancy and typically fall back to conventional retransmission once losses exceed the correction capability of the configured code. In contrast, BC-DIR adopts a repair strategy based on additional FEC symbols rather than retransmission of the original lost packets, enabling more effective deadline-constrained recovery under burst-loss conditions in V2X networks.

\section{System Model}\label{sec:system}

\subsection{V2X Network and FEC System Model}
\label{subsec:system_model}

\begin{figure}[t]
    \centering
    \includegraphics[width=1\columnwidth]{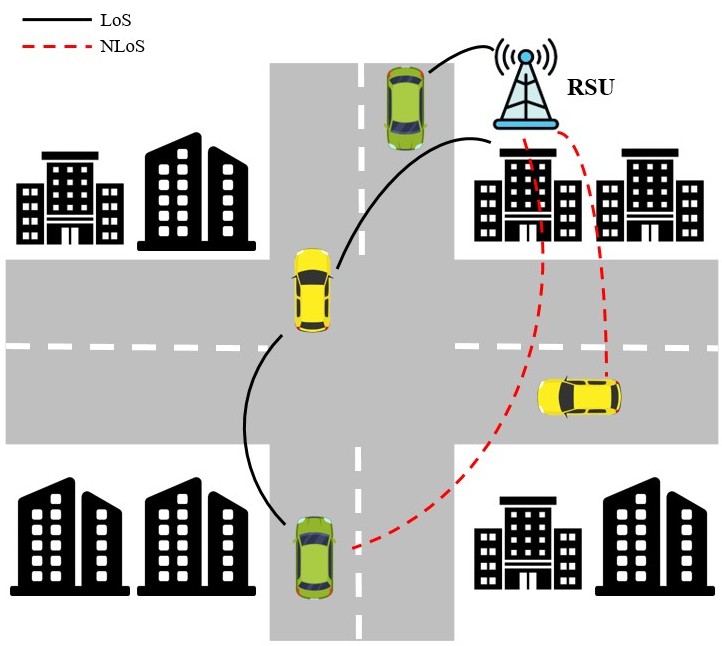}
    \caption{Urban crossroad V2X communication scenario.}
    \label{fig:v2x_scenario}
\end{figure}

We consider a V2X communication system operating in an urban environment characterized by complex road topology, high vehicle mobility, and frequent transitions between LoS and non-line-of-sight (NLoS) conditions, as illustrated in Fig.~\ref{fig:v2x_scenario}. Vehicles communicate with neighboring vehicles or roadside units (RSUs) over a shared wireless medium.

Due to building obstruction and traffic density variation at intersections, wireless impairments exhibit strong temporal correlation and typically manifest as burst losses. Such burst events persist over multiple consecutive packets and significantly impact deadline-constrained data delivery.

To improve reliability under burst losses, we employ systematic RaptorQ codes standardized in RFC~6330~\cite{rfc6330} as the underlying FEC mechanism. Application data is segmented into blocks consisting of $K$ source symbols. For each block, the encoder first transmits the $K$ systematic symbols, followed by repair symbols when redundancy is required. Let $R$ denote the total number of repair symbols transmitted, the total number of symbols transmitted for the block is given by
\[
N = K + R.
\]

At the receiver, decoding succeeds once at least $K+1$ innovative symbols are collected. In practice, successful decoding is achieved with high probability once $K+1$ symbols are received, with decoding failure probability below $10^{-4}$~\cite{lazaro2017inactivation}. This threshold-based property enables incremental redundancy: additional repair symbols can be injected until the decoding condition
is satisfied.

RaptorQ encoding and decoding complexities scale approximately linearly with block size, making it suitable for real-time V2X applications~\cite{rfc6330,lazaro2017inactivation}.

In this work, RaptorQ provides the block-level reliability primitive, while the proposed deadline-aware control framework determines when and how much additional redundancy should be transmitted to satisfy latency constraints under burst-loss conditions.

\subsection{End-to-End Delay Model}
\label{subsec:e2e_delay}

We now characterize the end-to-end delay of a message transmitted over a V2X link. Let $k$ denote the message index. The end-to-end delay of message $k$, denoted by $d_k$, is decomposed into two components:
\begin{equation}
  d_k = d_k^N + d_k^F,
\label{eq:e2edelay}
\end{equation}
where $d_k^N$ represents the network-layer delay incurred during message delivery, and $d_k^F$ captures the additional delay introduced by FEC-based loss recovery and decoding.

The network-layer delay $d_k^N$ aggregates transmission delay, medium access delay, queuing delay, processing delay, and propagation delay along the wireless path. Similar to prior work in mobile and wireless networks, this delay can be approximated by a one-way delay (OWD), which is approximately one-half of the measured RTT, i.e., $OWD \approx \mathrm{RTT}/2$.

Let $x_k$ denote the size of message $k$. Given an effective downlink transmission rate $C$, the wireless transmission delay can be approximated as $x_k / C$. The downlink rate is modeled as
\begin{equation}
  C = \xi W \log_2(1 + \sigma),
\end{equation}
where $W$ denotes the channel bandwidth, $\sigma$ is the received signal-to-interference-plus-noise ratio (SINR), and $\xi \in (0,1)$ captures implementation inefficiencies. Accordingly, the network-layer delay can be written as
\begin{equation}
  d_k^N \approx OWD + \frac{x_k}{C},
\end{equation}
which captures both propagation-related delay and transmission-related delay over the wireless downlink.

The term $d_k^F$ represents the FEC-induced delay arising from block-based recovery. A message is segmented into multiple FEC blocks. Under burst losses, some blocks may require additional repair rounds for feasible decoding, which increases the overall message completion time. The completion time of each message is given by
\begin{equation}
  d_k^F = d_k^C + d_k^D,
\end{equation}
where $d_k^C$ denotes the delay associated with collecting sufficient symbols for successful recovery of all FEC blocks belonging to message $k$, and $d_k^D$ denotes the corresponding FEC decoding delay.

\subsection{Completion Probability}
\label{subsec:completion_probability}

For deadline-constrained V2X services, the key performance objective is whether a message can be delivered before its transmission deadline. Let $T_k$ denote the transmission deadline associated with message $k$. Since each message is segmented into multiple FEC blocks, successful message delivery requires that all corresponding blocks be recovered and decoded within the deadline. Accordingly, the completion probability is defined as
\begin{equation}
P_k
=
\Pr\!\left\{
d_k \le T_k
\right\}.
\label{eq:completion_probability}
\end{equation}
In the following, we further analyze this completion event through the block-level recovery process under burst-loss channel, and derive the corresponding deadline-constrained reliability model in Section~\ref{subsubsec:deadline_reliability}.

\subsection{Transmission Overhead}
\label{subsec:transmission_overhead}

In deadline-constrained V2X transport, redundancy overhead is not limited to repair symbols alone. In addition to FEC repair data, extra transmission cost also arises from protocol headers, control feedback such as NACK messages, and other auxiliary signaling required for reliable recovery.

Let $D_k$ denote the application data size of message $k$, and let $H_k$ denote the total non-payload overhead incurred during the transmission of message $k$, including repair symbols, FEC-related headers, and control messages. The transmission overhead is defined as
\begin{equation}
o_k = \frac{H_k}{D_k}.
\label{eq:overhead_def}
\end{equation}

\section{Bandit-Controlled Deadline-Aware Incremental Redundancy}
\label{sec:bandit_ir_quic}
In this section, we present the design of BC-DIR, which integrates a deadline-aware incremental redundancy mechanism with a bandit-based online control strategy to adapt redundancy configurations from end-to-end transport feedback. BC-DIR builds on the rateless property of RaptorQ to support deadline-aware incremental redundancy in V2X transport. After a decoding failure, recovery is performed by transmitting newly generated repair symbols rather than directly retransmitting the original lost packets. By injecting additional repair symbols in each recovery round, DIR increases the probability of successful decoding under burst loss while preserving flexibility in redundancy allocation.

QUIC remains responsible for congestion control. 
Following RFC~9265~\cite{rfc9265}, loss events are always exposed to congestion control even if they are later repaired by FEC, and repair packets are subject to the same congestion-control constraints as original data packets. 
BC-DIR is implemented over QUIC datagrams~\cite{rfc9221}, with BBRv2 adopted as the underlying congestion controller because it is better suited to wireless loss environments than purely loss-driven schemes~\cite{michel2024quirl}.

\subsection{Online Redundancy Adaptation Problem}

FEC improves the completion probability $P_k$ by increasing the likelihood that all FEC blocks of message $k$ can be recovered before the deadline. At the same time, additional repair symbols increase the transmission overhead $o_k$. Since the reliability gain of redundancy typically exhibits diminishing returns, redundancy control is naturally formulated as a tradeoff between completion probability and overhead.

A fixed or offline-configured redundancy strategy is generally inadequate in the considered V2X scenario, because burst loss, mobility-induced channel variation, and medium-access contention vary over time. Accordingly, the redundancy configuration should be adjusted according to the changing network condition. This leads to an online sequential decision problem in which an action is selected at the beginning of each round and its transport outcome is then observed through end-to-end feedback.

Reinforcement learning (RL) provides a general framework for such sequential decision problems~\cite{zhu2023transfer}. However, a full RL formulation is most useful when actions significantly affect future state transitions and long-term action-dependent dynamics must be explicitly modeled. In the V2X scenario, the dominant dynamics are mainly driven by external factors, and are only weakly influenced by the selected redundancy action. The action therefore mainly affects the immediate reward of the current round, rather than the future network state.

For this reason, a contextual bandit formulation is more suitable than a full Markov decision process. Instead of learning long-horizon state transitions, the controller only needs to learn which redundancy action yields the highest expected reward under the currently observed context~\cite{bubeck2012regret}. Motivated by this problem structure, we adopt a contextual multi-armed bandit to adapt redundancy configurations from end-to-end transport feedback.

\subsection{Deadline-Aware Incremental Redundancy Mechanism}
\label{subsec:ir_fec_mechanism}

\begin{figure}[t]
    \centering
    \includegraphics[width=1\columnwidth]{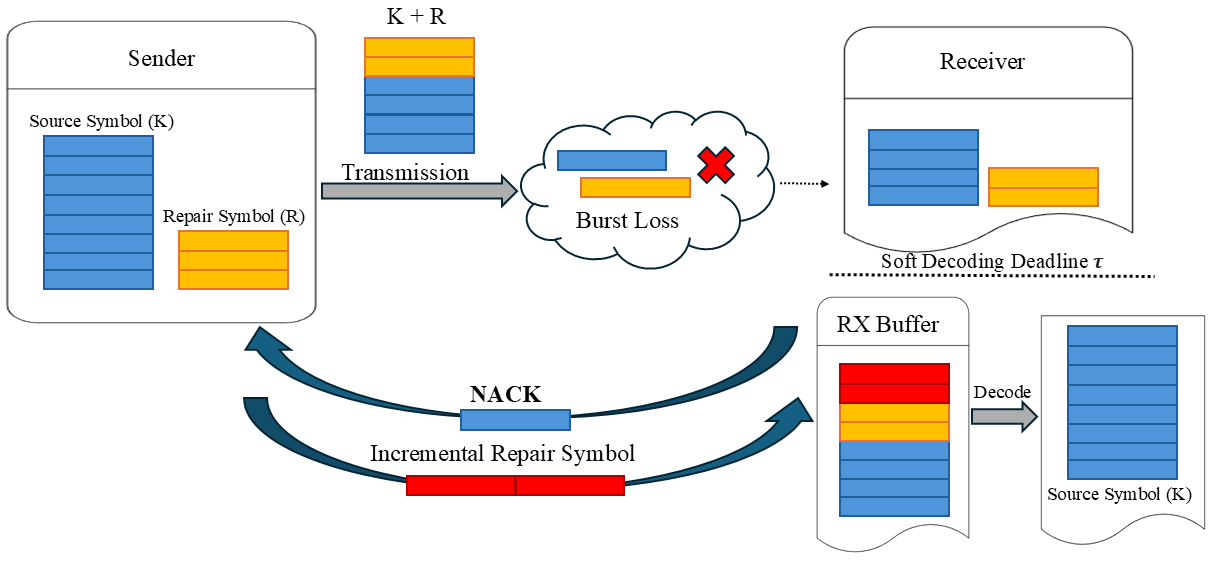}
    \caption{Block-level incremental-redundancy FEC transmission and decoding process with a soft decoding deadline.}
    \label{fig:incremental_ir}
\end{figure}

Fig.~\ref{fig:incremental_ir} illustrates the proposed DIR FEC transmission mechanism. At the sender, application data is first segmented into consecutive blocks, each containing $K$ source symbols. For every block, the encoder generates $R$ repair symbols and transmits the resulting $K + R$ encoded symbols over the network. Due to wireless losses, the receiver may obtain only a subset of these symbols.

The receiver buffers all successfully received symbols on a per-block basis. As symbols arrive, the receiver continuously checks whether the decoding condition is satisfied. If enough innovative symbols are available, the receiver immediately performs FEC decoding and reconstructs the original data block.

Each block is associated with a soft decoding deadline $\tau$. If decoding has not succeeded before $\tau$, the receiver assumes that additional redundancy is required and sends a NACK to the sender requesting repair symbols.

Upon receiving a NACK, the sender gets the number of missing symbols based on receiver feedback and generates additional repair symbols. Specifically, the sender transmits $d$ repair symbols to compensate for the observed deficit, together with an additional redundancy margin $\Delta R$ to protect against further losses in the subsequent transmission round.

To prevent overly frequent retransmission triggering, a cooldown interval is introduced after each repair round. Once a NACK is sent, the receiver waits for approximately $\tau + RTT$ before issuing another repair request. This waiting interval allows the newly transmitted repair symbols to arrive and avoids excessive NACK signaling. The process repeats until decoding succeeds or the block deadline expires.

\subsection{End-to-End Workflow of BC-DIR}
\label{subsec:payload_structure}
\begin{figure*}[!t]
    \centering
    \includegraphics[width=2\columnwidth]{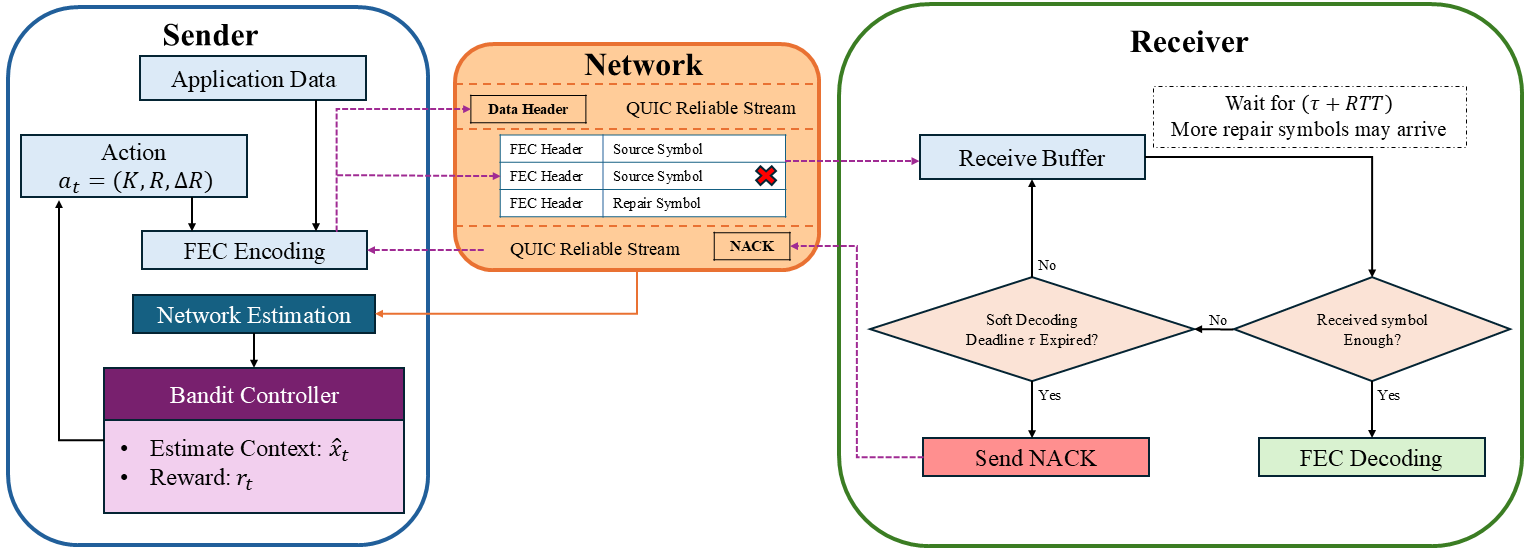}
    \caption{End-to-end workflow of BC-DIR. Control messages are delivered through QUIC reliable streams, while source and repair symbols are transmitted as QUIC datagrams and may experience packet loss.}
    \label{fig:workflow}
\end{figure*}

\begin{table}[t]
\centering
\caption{Data Header fields for a transmission round.}
\label{tab:dataheader}
\begin{tabular}{l c p{5cm}}
\hline
\textbf{Name} & \textbf{Size (B)} & \textbf{Description} \\
\hline
V     & 1 & Protocol version identifier. \\
$D$    & 8 & Total application data size in the current transmission round. \\
$L$   & 4 & Payload size of each source symbol. \\
$\tau$ & 4 & Soft decoding deadline per block. \\
\hline
\end{tabular}
\end{table}

Fig.~\ref{fig:workflow} illustrates the end-to-end workflow of BC-DIR. The protocol separates control signaling from data transmission. Control messages are delivered through QUIC’s reliable streams, while source and repair symbols are transmitted using QUIC datagrams and may experience packet loss.

Each transmission round begins with a \emph{Data Header} sent through QUIC’s reliable transport. This message establishes the configuration for the current round, including the total application data size $D$, the symbol payload size $L$, and the soft decoding deadline $\tau$. Because the Data Header is carried on the reliable streams, both endpoints maintain a consistent configuration state.

After the Data Header exchange, application data is transmitted as FEC-protected symbols using QUIC datagrams. Each datagram carries an FEC header followed by a payload of size $L$, and the corresponding header fields are summarized in Table~\ref{tab:fecheader}. Symbols with the same \texttt{BlockID} belong to the same FEC block and are buffered together at the receiver, while \texttt{SymID} identifies the encoding symbol within the block. Unlike control messages, these datagrams may be lost in the wireless channel.

As symbols arrive, the receiver continuously checks whether the current block can be decoded. If sufficient symbols have been collected, decoding is immediately performed. Otherwise, the receiver waits until the soft decoding deadline $\tau$ expires. If decoding is still infeasible at that time, the receiver sends a NACK carrying the corresponding \texttt{BlockID}. NACK messages are also transmitted through QUIC’s reliable streams so that control feedback is never lost.

Upon receiving a NACK, the sender generates additional repair symbols for the requested block and transmits them as new datagrams. To avoid overly frequent repair triggering, the receiver waits for approximately $\tau+RTT$ after issuing a NACK before sending another request, allowing the newly transmitted symbols to arrive and be incorporated into the decoding buffer.

This receiver-driven incremental redundancy process repeats until the block is successfully decoded or the block deadline expires.
Once all blocks in the current transmission round have been decoded and the total data size $D$ has been delivered,
the receiver sends a completion notification to the sender.
A new round begins only after this completion signal is received, at which point a new Data Header is issued.

\begin{table}[t]
\centering
\caption{FEC header fields carried by each transmitted symbol.}
\label{tab:fecheader}
\begin{tabular}{l c p{5cm}}
\hline
\textbf{Name} & \textbf{Size (B)} & \textbf{Description} \\
\hline
V & 1 & Protocol version identifier. \\
BlockID & 2 & Identifier of the FEC block. \\
$N$       & 1 & Total number of symbols transmitted so far for the block. \\
$K$       & 1 & Number of source symbols in the block. \\
SymID   & 1 & Encoding symbol identifier within the block. \\
\hline
\end{tabular}
\end{table}

\subsection{Soft Decoding Deadline Design}
\label{subsec:tau_design}

The parameter $\tau$ is used as a soft decoding deadline for receiver-side repair triggering. It specifies how long the receiver should continue waiting for symbols from the current transmission round before requesting additional repair symbols. Its design is particularly important in V2X environments, where burst losses are often accompanied by packet reordering. In such cases, relying solely on ACK-driven loss detection may trigger repair too aggressively, because out-of-order delivery can be misinterpreted as actual loss.

Therefore, $\tau$ serves as a short reassembly window after the nominal transmission span of the current round.  A small $\tau$ tends to trigger repair prematurely, while an excessively large $\tau$ delays repair activation and reduces the remaining time budget available for subsequent recovery rounds. To balance these effects, $\tau$ is defined from the FEC block length of the current transmission round.

Let $\Delta$ denote the nominal symbol transmission interval. Under QUIC, packet transmission is paced by the congestion controller, so $\Delta$ is determined by the current sending rate and reflects the nominal time consumed by one transmitted symbol~\cite{rfc9002}. For a round consisting of $N$ transmitted symbols, we define
\begin{equation}
\tau = \rho N \Delta,
\label{eq:tau_design}
\end{equation}
where $\rho\in(0,1)$ is a proportional waiting factor.

This definition makes the soft decoding deadline proportional to the scale of the current transmission round. 
It allows delayed symbols to arrive before repair is triggered, thereby reducing unnecessary feedback under burst loss and packet reordering, while still preserving sufficient time for additional recovery attempts before the overall deadline expires.

\subsection{Deadline-Constrained Reliability of ARQ and DIR}
\label{subsubsec:deadline_reliability}

We analyze deadline-constrained reliability of ARQ and DIR under a Gilbert-Elliott (GE) burst-loss channel.

At symbol time $t$, the channel state $s_t \in \{\mathrm G,\mathrm B\}$ follows a two-state Markov chain with transition probabilities $\alpha$ and $\beta$. Symbols transmitted in the Good state are successfully received, whereas symbols transmitted in the Bad state are lost. Let $S_N$ denote the number of lost symbols among $N$ transmissions. Under the GE model,
\[
S_N = \sum_{t=1}^{N} \mathbf{1}\{s_t=\mathrm B\}.
\]

Decoding succeeds once at least $K+1$ innovative symbols are collected. Let each triggered repair round incur a cooldown $c_0 = RTT+\tau$.

\subsubsection{Initial Phase and Conditional Deficit}

After transmitting $N_0=K+R$ symbols, initial decoding succeeds if
\begin{equation}
S_{N_0}\le R-1.
\end{equation}
The corresponding success probability under GE loss is
\begin{equation}
p_0 = \Pr(S_{N_0}\le R-1),
\end{equation}
where $S_{N_0}$ follows a Markov-binomial distribution.

Conditioned on initial failure, let $d$ denote the innovative deficit. In the following analysis, ARQ and DIR are compared under the same realized deficit $d$, so that both schemes operate under identical loss realizations.

\subsubsection{GE-Based Symbol Repair Model}

Upon a repair trigger with deficit $d$, the sender transmits $d+\Delta R$ symbols, where $\Delta R\ge 0$ is the redundancy margin. Under the GE channel, the repair round succeeds if and only if
\begin{equation}
S_{d+\Delta R}\le \Delta R.
\label{eq:round_success_strict}
\end{equation}
The one-round success probability is therefore
\begin{equation}
p(\Delta R)
=
\Pr(S_{d+\Delta R}\le \Delta R),
\label{eq:pDeltaR_GE}
\end{equation}
where $S_{d+\Delta R}$ is the GE loss count over the repair window. ARQ corresponds to $\Delta R=0$, and DIR corresponds to $\Delta R>0$. The following lemma formalizes a basic structural property of this mechanism: increasing the redundancy margin $\Delta R$ cannot decrease the probability of one-round successful repair.

\begin{lemma}\label{lem:monotone_p}

Under the GE loss model, the one-round success probability
$p(\Delta R)$ defined in~\eqref{eq:pDeltaR_GE}
is non-decreasing in $\Delta R$, i.e.,
\[
p(\Delta R+1)\ \ge\ p(\Delta R),\qquad \forall \Delta R\ge 0.
\]

\end{lemma}

\begin{proof}

Let $n=d+\Delta R$.
For any realization of the channel state sequence, $S_n$ is non-decreasing in $n$. Consider the margin $\Delta R$, the event
\[
\{S_{d+\Delta R}\le \Delta R\},
\]
implies that at most $\Delta R$ losses occur in the first $d+\Delta R$ transmissions. When extending to $d+\Delta R + 1$ transmissions, the allowable loss threshold increases to $\Delta R + 1$, which enlarges the admissible loss region. Hence
\[
\{S_{d+\Delta R}\le \Delta R\}
\subseteq
\{S_{d+\Delta R+1}\le \Delta R + 1\},
\]
which implies
\[
p(\Delta R)\le p(\Delta R+1).
\]

\end{proof}

\subsubsection{Deadline-Constrained Comparison Between DIR and ARQ}

After the initial transmission phase, the remaining deadline budget limits how many repair rounds can still be executed. Under a given realized deficit $d$, let $V(\Delta R)$ denote the maximum number of feasible repair rounds under redundancy margin $\Delta R$, ARQ and DIR therefore differ in two coupled aspects: the success probability of each repair round and the number of repair opportunities available before the deadline. This leads to the following sufficient condition under which DIR achieves strictly higher deadline-constrained success probability than ARQ.

\begin{theorem}
\label{thm:DIR_vs_ARQ_strict}
For any $\Delta R>0$, if the repair process is feasible under the deadline constraint and
\begin{equation}
p(\Delta R)
>
1-(1-p(0))^{\frac{V(0)}{V(\Delta R)}},
\tag{C1}
\label{eq:C1_strict}
\end{equation}
then
\[
P_{\mathrm{DIR}}(T;\Delta R\mid d)
>
P_{\mathrm{ARQ}}(T\mid d).
\]
\end{theorem}

\begin{proof}
After the initial phase, the sender has already transmitted $K+R$ symbols. Since each symbol consumes transmission time $\Delta$, the remaining transmission time is
\[
T_{\mathrm{eff}} = T-(K+R)\Delta.
\]
For a repair round with redundancy margin $\Delta R$, the sender transmits $d+\Delta R$ symbols, and each triggered round also incurs a cooldown $c_0$. Hence the total time consumed by one repair round is
\[
c_0 + (d+\Delta R)\Delta.
\]
Therefore, the maximum number of feasible repair rounds under the remaining deadline constraint is:
\[
V(\Delta R)
=
\left\lfloor
\frac{T_{\mathrm{eff}}}
{c_0 + (d+\Delta R)\Delta}
\right\rfloor.
\]
To enable at least one repair opportunity, $\Delta R$ must satisfy
\[
c_0 + (d+\Delta R)\Delta \le T_{\mathrm{eff}}.
\]
Hence the feasible redundancy set is finite:
\begin{equation}
    \Delta R_{\max} = \left\lfloor\frac{T_{\mathrm{eff}} - c_0 - d\Delta}{\Delta}\right\rfloor.\label{eq:r_max}
\end{equation}

The deadline-constrained block success probability under redundancy margin $\Delta R$ is
\[
P(T;\Delta R\mid d)
=
p_0+(1-p_0)\Big(1-(1-p(\Delta R))^{V(\Delta R)}\Big).
\]
Accordingly,
\begin{align*}
&P_{\mathrm{ARQ}}(T\mid d)=P(T;0\mid d),\\
&P_{\mathrm{DIR}}(T;\Delta R\mid d)=P(T;\Delta R\mid d).
\end{align*}
Since the term $p_0$ is common to both schemes, comparing DIR and ARQ reduces to comparing the probability that at least one repair round succeeds. Thus,
\[
P_{\mathrm{DIR}}(T;\Delta R\mid d)
>
P_{\mathrm{ARQ}}(T\mid d)
\]
holds if and only if
\[
1-(1-p(\Delta R))^{V(\Delta R)}
>
1-(1-p(0))^{V(0)}.
\]
This is equivalent to
\[
(1-p(\Delta R))^{V(\Delta R)}
<
(1-p(0))^{V(0)}.
\]
Since $V(\Delta R)\ge 1$, taking the $V(\Delta R)$-th root on both sides yields
\[
1-p(\Delta R)
<
(1-p(0))^{\frac{V(0)}{V(\Delta R)}},
\]
or equivalently,
\[
p(\Delta R)
>
1-(1-p(0))^{\frac{V(0)}{V(\Delta R)}}.
\]
\end{proof}

\begin{remark}[Practical condition for DIR advantage]\label{remark:engineering_regime_final}

In practical wireless systems, the cooldown time $c_0$ is typically on the order of milliseconds, whereas symbol transmission time $\Delta$ is on the order of microseconds. Hence
\[
(d+\Delta R)\Delta \ll c_0,
\]
so $V(\Delta R)$ is dominated by $c_0$
and usually satisfies
\[
V(\Delta R)\in\{V(0),\,V(0)-1\}.
\]
In the dominant case $V(\Delta R)=V(0)-1$, condition~\eqref{eq:C1_strict} reduces to
\[
p(\Delta R)
>
1-(1-p(0))^{\frac{V(0)}{V(0)-1}}.
\]

In practice, even moderate redundancy margins substantially increase $p(\Delta R)$, while reducing $V(\Delta R)$ by at most one round. Therefore, for typical timing and burst-loss parameters, DIR achieves higher deadline-constrained reliability than ARQ in most operating regimes, although~\eqref{eq:C1_strict} is not universally tight.
\end{remark}

\subsubsection{Existence of an Optimal Redundancy Margin}
\begin{theorem}\label{thm:existence_optimal_margin_clean}

Under the deadline constraint, there exists at least one
\[
\Delta R^\star\in\{0,1,\dots,\Delta R_{\max}\}
\]
such that
\[
P(\Delta R^\star)
=
\max_{\Delta R\in\{0,\dots,\Delta R_{\max}\}}
P(\Delta R).
\]

\end{theorem}

\begin{proof}

The feasible redundancy set is finite by~\eqref{eq:r_max}. Since $P(\Delta R)$ is well-defined for each admissible $\Delta R$, a maximizer exists.

\end{proof}

The existence of an optimal redundancy margin $\Delta R^\star$ follows from the deadline-imposed upper bound on admissible redundancy. However, $\Delta R^\star$ depends on channel statistics through $p_0$ and $p(\Delta R)$, which in turn depend on the underlying loss dynamics and timing parameters.

In dynamic V2X environments, these statistics are generally unknown and time-varying, and cannot be reliability estimated in closed form. Consequently, the optimal repair margin cannot be computed a priori and must be learned online. This motivates the introduction of a bandit-based adaptive controller to track the optimal $\Delta R$ under evolving network conditions.

\subsection{Contextual Bandit Control}
\label{subsec:bandit_formulation}

BC-DIR adopts a contextual bandit to adapt the DIR configuration on a per-round basis.
Each transmission round $t$ corresponds to the delivery and recovery of a set of
FEC-protected blocks. At the beginning of round $t$, the sender selects a DIR
configuration according to recent end-to-end feedback, without explicitly modeling
channel state transitions.

\subsubsection{Context}
Let $\mathbf{x}_t \in \mathbb{R}^6$ denote the observable context at round $t$:
\begin{equation}
\mathbf{x}_t = (\phi_t,\; o_t,\; \gamma_t,\; \eta_t,\; \Lambda_t,\; \Theta_t),
\label{eq:context_def_compact}
\end{equation}
where $\phi_t$ and $o_t$ are the achieved goodput and transmission overhead,
$\gamma_t$ is the number of NACK-triggered repair rounds,
$\eta_t\in\{0,1\}$ indicates whether the transmission round succeeds,
\begin{equation}
\Lambda_t=\frac{R_t}{K_t},
\label{eq:fec_rate_compact}
\end{equation}
is the effective FEC rate, and
\begin{equation}
\Theta_t=\frac{1}{W}\sum_{i=t-W}^{t}\mathbf{1}\{\eta_i=0\},
\label{eq:theta_def_compact}
\end{equation}
is the recent failure ratio over a window of size $W$.

To suppress short-term fluctuations caused by burst losses and MAC contention,
EWMA smoothing is applied to selected metrics. For a generic metric $z_t\in\{\phi_t,o_t,\gamma_t,\eta_t\}$,  the refined estimate $\hat{z}_t$ is updated recursively as
\begin{equation}
\hat{z}_t=\lambda \hat{z}_{t-1}+(1-\lambda)z_t,
\qquad
\lambda\in(0,1),
\label{eq:ewma_compact}
\end{equation}
where $\lambda$ controls the smoothing factor. The final refined context $\hat{\mathbf{x}}_t=(\hat{\phi}_t,\hat{o}_t,\hat{\gamma}_t,\hat{\eta}_t,\Lambda_t,\Theta_t)$.

\subsubsection{Action space}
Each action corresponds to one admissible DIR configuration,
\begin{equation}
a=(K,\;R,\;\Delta R),
\label{eq:action_def_compact}
\end{equation}
where the tuple determines the block size, initial redundancy,
incremental repair margin, and soft decoding deadline.
Let $\mathcal{A}$ denote the finite set of admissible actions.
At round $t$, the sender observes $\hat{\mathbf{x}}_t$ and selects
$a_t\in\mathcal{A}$, which remains fixed during that round.

\subsubsection{Reward}
The reward is designed to balance timely delivery, reliability, feedback cost, and redundancy overhead. It consists of four components:
\begin{align*}
r_t^\phi &= w_{\phi}\,\frac{\phi_t}{c},\\
r_t^\eta &= -\,w_{\eta}(1-\eta_t),\\
r_t^\gamma &= -\,w_{\gamma}\gamma_t,\\
r_t^o &= w_o\frac{1}{1+\left(\frac{o_t}{\delta}\right)^2},
\end{align*}
where $c$ is the nominal link capacity and $\delta$ controls
the sensitivity of the overhead penalty.
The total reward is
\begin{equation}
r_t = r_t^\phi + r_t^\eta + r_t^\gamma + r_t^o.
\label{eq:reward_total_compact}
\end{equation}

This reward favors configurations that deliver useful data before the deadline with fewer repair rounds and moderate redundancy, enabling online adaptation to time-varying V2X channel conditions.

\subsection{Linear Context-Action Modeling and Thompson Sampling}
\label{subsec:lints_model}

For each transmission round $t$, the sender observes the refined context
$\hat{\mathbf{x}}_t$ and selects a DIR configuration $a_t$.
The expected reward is modeled as a linear function of a joint
context-action feature vector:
\begin{equation}
\mathbb{E}[r_t \mid \hat{\mathbf{x}}_t,a_t]
=
\boldsymbol{f}(\hat{\mathbf{x}}_t,a_t)^\top \boldsymbol{\theta},
\label{eq:lints_linear_model_compact}
\end{equation}
where $\boldsymbol{\theta}$ is an unknown parameter vector.

The feature vector is constructed as
\begin{equation}
\boldsymbol{f}(\hat{\mathbf{x}}_t,a_t)
=
\begin{bmatrix}
1 \\
\hat{\mathbf{x}}_t \\
\boldsymbol{\psi}(a_t) \\
\hat{\mathbf{x}}_t \otimes \boldsymbol{\psi}(a_t)
\end{bmatrix},
\label{eq:lints_feature_compact}
\end{equation}
where $\boldsymbol{\psi}(a_t)$ is the encoding of the discrete action
and $\otimes$ denotes the Kronecker product.
This representation captures context effects, action effects,
and their interactions in a unified linear model.

BC-DIR uses Linear Thompson Sampling (LinTS) to maintain a Gaussian posterior
over $\boldsymbol{\theta}$:
\begin{equation}
\boldsymbol{\theta}
\sim
\mathcal{N}\!\left(
\hat{\boldsymbol{\theta}}_{t-1},
\mathbf{V}_{t-1}^{-1}
\right),
\end{equation}
where $\hat{\boldsymbol{\theta}}_{t-1}$ is the posterior mean and
$\mathbf{V}_{t-1}$ is the precision matrix.

At round $t$, a parameter sample $\tilde{\boldsymbol{\theta}}_t$
is drawn from the posterior, and the action is selected as
\begin{equation}
a_t
=
\arg\max_{a\in\mathcal{A}}
\boldsymbol{f}(\hat{\mathbf{x}}_t,a)^\top
\tilde{\boldsymbol{\theta}}_t.
\label{eq:lints_action_compact}
\end{equation}

After observing reward $r_t$, the posterior statistics are updated by Bayesian linear regression with exponential discounting:
\begin{align}
\mathbf{V}_t
&=
\lambda \mathbf{V}_{t-1}
+
\boldsymbol{f}(\hat{\mathbf{x}}_t,a_t)
\boldsymbol{f}(\hat{\mathbf{x}}_t,a_t)^\top,
\label{eq:lints_precision_update_compact}
\\
\mathbf{b}_t
&=
\lambda \mathbf{b}_{t-1}
+
r_t\,\boldsymbol{f}(\hat{\mathbf{x}}_t,a_t),
\label{eq:lints_b_update_compact}
\end{align}
where $\lambda\in(0,1)$ discounts older observations.
The posterior mean is then updated as
\begin{equation}
\hat{\boldsymbol{\theta}}_t
=
\mathbf{V}_t^{-1}\mathbf{b}_t.
\label{eq:lints_mean_update_compact}
\end{equation}

This online update enables BC-DIR to adapt its redundancy configuration to non-stationary V2X channel conditions while keeping the control logic lightweight.

Overall, BC-DIR improves deadline-constrained completion probability by trading moderate redundancy overhead for higher per-round repair success. Compared with pure retransmission, the proposed incremental-redundancy mechanism reduces repeated feedback cycles and increases the likelihood of successful decoding within a limited number of repair opportunities. The analytical results further show that an appropriate redundancy margin exists under deadline constraints, but its optimal value depends on time-varying channel and timing conditions. This motivates the bandit controller, which adaptively adjusts the redundancy configuration online according to end-to-end observations, enabling BC-DIR to balance reliability, latency, and bandwidth efficiency in dynamic V2X environments.

\section{Evaluation}
\label{sec:evaluation}

In this section, we evaluate BC-DIR from three complementary perspectives. We first validate the analytical insights through Monte Carlo simulations under burst-loss channels, focusing on the reliability advantage of DIR over ARQ and the role of the redundancy margin under deadline constraints. We then conduct an ablation study to isolate the contributions of initial FEC protection and incremental repair in the complete BC-DIR design. Finally, we evaluate the full system in realistic vehicular simulations under both congested urban V2X scenarios and favorable network conditions, examining completion ratio, overhead, worst-case delay, and recovery rounds across different traffic intensities.
\subsection{Experiment Setup}

\subsubsection{Implementation and Network Configuration}

The proposed BC-DIR mechanism is implemented on top of the official \texttt{quic-go} stack, which is one of the most widely deployed open-source implementations of QUIC in practical Internet services. Since BBRv2 estimates bottleneck bandwidth and propagation delay rather than reacting directly to packet loss, it is particularly suitable for V2X wireless environments where losses are frequently caused by channel impairments instead of persistent congestion.

The primary evaluation scenario is implemented using the Veins framework integrated with OMNeT++, with SUMO providing microscopic vehicular mobility~\cite{sommer2011bidirectionally}. The wireless link follows the IEEE 802.11p protocol~\cite{ieee80211p}, which is widely adopted in vehicular communications. The physical and MAC parameters used in the mobility-driven simulations are summarized in Table~\ref{tab:sim_params}.

\subsubsection{Vehicular Mobility and Evaluation Scenarios}
\begin{table}[t]
\centering
\caption{Vehicular Network Simulation Parameters}
\label{tab:sim_params}
\begin{tabular}{l l}
\hline
Parameter & Value \\
\hline
Maximum interference distance & 2600 m \\
Transmission power & 20 mW \\
PHY bitrate & 10 Mbps \\
Minimum reception power & $-110$ dBm \\
Noise floor & $-98$ dBm \\
Number of RSUs & 1 \\
Maximum vehicle speed & 60 km/h \\
Traffic intensity & 300, 600, 900, 1200 vehicles/hour \\
\hline
\end{tabular}
\end{table}

The mobility-driven evaluation scenario represents a dense urban intersection with two lanes per direction and surrounding buildings that introduce frequent non-line-of-sight conditions and signal obstruction. Vehicles periodically communicate with a single roadside unit (RSU) and exchange application data under shared wireless access. The traffic intensity is varied from 300 to 1200 vehicles per hour to evaluate performance under different congestion levels~\cite{petrov2021performance}. The resulting channel conditions exhibit significant burst losses and jitter, which are characteristic of real-world V2X environments.

A single RSU is used to isolate the impact of deadline-constrained loss recovery. With multiple RSUs, vehicles may switch their access point during mobility, introducing additional performance variation from association changes, MAC-layer contention redistribution, and transient link fluctuations. Such effects would couple transport-layer recovery performance with access switching behavior, making it difficult to attribute the observed gain specifically to the proposed redundancy mechanism. The single-RSU setting therefore provides a controlled scenario for evaluating how BC-DIR responds to burst loss and congestion under dynamic vehicular traffic.

In addition to the above mobility-driven scenario, we consider a second evaluation setting to isolate redundancy-control behavior under predominantly line-of-sight and favorable communication conditions. This setting serves as a reference case for high-quality V2X links, where packet loss is low but non-zero and typically remains below 1\%~\cite{eckermann2019performance}.

\subsubsection{Traffic Types and Data Objects}

The evaluation focuses on medium-scale data objects for which forward error correction can provide meaningful benefits. Very small safety beacons such as BSM or CPM packets are excluded~\cite{kenney2011dsrc}, since their payload sizes are insufficient to amortize FEC overhead.

Two representative data scales are considered: (i) approximately 100~KB objects corresponding to local dynamic map updates or perception data slices~\cite{etsi_tr_103_562}, and (ii) approximately 1~MB objects representing high-resolution perception segments, point cloud sharing, or OTA update fragments~\cite{3gpp_tr_22_886}.

\subsubsection{Benchmark Schemes}

BC-DIR is compared against the following baselines:

\begin{itemize}

\item \textbf{Native QUIC}: The baseline transport is the official \texttt{quic-go} implementation without any FEC mechanism. It represents the most widely used real-world QUIC deployment and relies purely on loss detection and retransmission for reliability.

\item \textbf{FLEC}~\cite{michel2023flec}: FLEC is implemented following its PQUIC-based design and employs a RLC FEC scheme. It dynamically adjusts the coding rate according to observed network conditions.

\item \textbf{Ablation Baselines}: To separately evaluate the contributions of initial FEC protection and deadline-aware incremental repair, we further include two fixed-parameter ablation baselines without bandit-based adaptation.

FEC-only: $(K=40, R=10, \Delta R=0)$. 
It provides only proactive redundancy in the initial transmission and disables additional repair-symbol injection after decoding failure.

DIR-only: $(K=40, R=0, \Delta R=10)$. 
It disables proactive redundancy in the initial transmission and relies entirely on incremental repair after decoding failure.

\end{itemize}

All schemes share identical congestion control, transport configuration, and bandwidth constraints, ensuring a fair comparison.

\subsubsection{Bandit Control Parameter Space}
The action space consists of $K\in[20,60]$, $R\in[0,20]$, $\Delta R\in[0,20]$.

The maximum repair margin is set to $\Delta R_{\max}=20$. As shown in Fig.~\ref{fig:vary_GE}, the reliability gain becomes marginal once $\Delta R$ exceeds 20 under typical burst-loss conditions. Increasing $\Delta R$ beyond this point yields limited improvement.

The bandit reward function balances completion probability, 
latency satisfaction, retransmission efficiency, and redundancy overhead.
The weight parameters used in training are:
$w_{\phi}=1$, 
$w_{\eta}=0.3$, 
$w_{\gamma}=0.3$, 
$w_{o}=0.3$, 
with exploration parameter $\delta=0.25$. These values are fixed across all experiments.



\subsection{Validation of Analytical Insights}\label{subsec:theory_validation}


We evaluate the deadline-constrained reliability via Monte Carlo simulation over a two-state GE loss model. The experiments only simulate the symbol transmission and loss process, and estimate the corresponding recovery success probability statistically. For each operating point, the experiment is independently repeated 5000 times. Unless otherwise specified, the parameter setting is aligned with the representative V2X configuration, with bandwidth $=10$~Mbps, $RTT=100$~ms, and block size $=100$~KB. To isolate the behavior of the repair mechanism, the initial redundancy is set to $R=0$. The default GE parameters are $(\alpha,\beta)=(0.03,0.3)$, corresponding to a typical burst-loss regime in vehicular wireless links.

\subsubsection{DIR vs. ARQ under Burst Loss}

\begin{figure}[t]
\centering
\includegraphics[width=0.9\columnwidth]{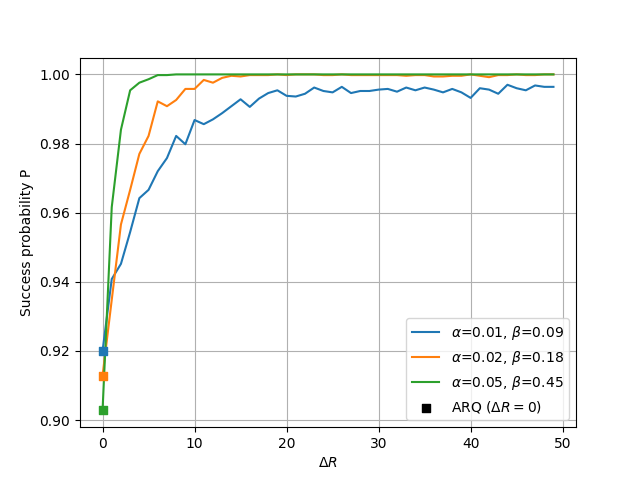}
\caption{Effect of channel burstiness on ARQ and DIR reliability under the GE model.}
\label{fig:vary_GE}
\end{figure}
Fig.~\ref{fig:vary_GE} shows the success probability as a function of the redundancy margin $\Delta R$ under different GE parameters. ARQ corresponds to $\Delta R=0$.

Across all tested burst configurations, any positive redundancy margin increases reliability compared with ARQ. The performance gap widens under longer burst conditions (smaller $\beta$), where retransmission-only recovery becomes increasingly inefficient due to repeated loss within burst periods.  Injecting additional innovative symbols per repair round increases the probability that at least one round successfully closes the deficit, leading to consistently higher deadline-constrained reliability.

These results validate the structural comparison established in Theorem~\ref{thm:DIR_vs_ARQ_strict}: under identical deficit conditions, DIR achieves no worse and typically strictly higher reliability than ARQ in burst-dominated regimes.

\subsubsection{Optimal Redundancy Margin under Deadline Constraints}
\begin{figure}[t]
\centering
\includegraphics[width=0.9\columnwidth]{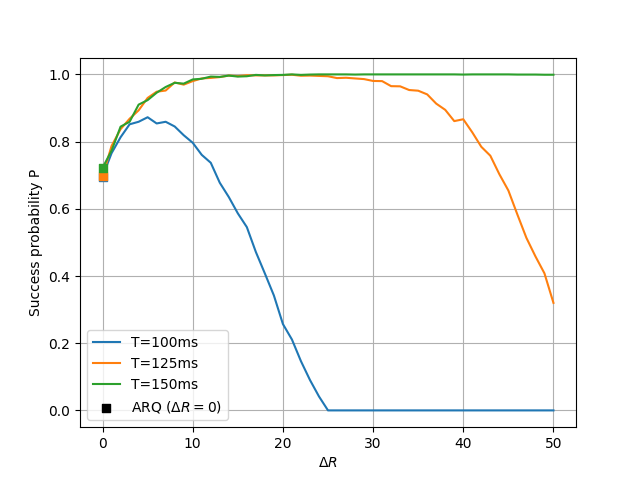}
\caption{Block success probability versus redundancy margin under different deadline constraints.}
\label{fig:vary_deadline}
\end{figure}

Fig.~\ref{fig:vary_deadline} presents the success probability versus $\Delta R$ under different deadline constraints ($T=100$~ms, $125$~ms, $150$~ms).

Under a tight deadline ($T=100$~ms), the curve exhibits a clear unimodal behavior. Increasing $\Delta R$ initially improves repair success probability, while excessive redundancy reduces the number of feasible repair rounds, causing reliability to decrease. This behavior directly reflects the deadline-induced repair budget tradeoff characterized in Theorem~\ref{thm:existence_optimal_margin_clean}.

As the deadline relaxes, the feasible repair budget expands and the decreasing region gradually diminishes. For sufficiently relaxed deadlines, reliability approaches monotonic saturation, indicating that time budget constraints become less dominant.

The results empirically confirm the existence of a finite optimal redundancy margin under tight deadline conditions.

\subsubsection{Impact of Concurrent Block Transmission}
\begin{figure}[t]
\centering
\includegraphics[width=0.9\columnwidth]{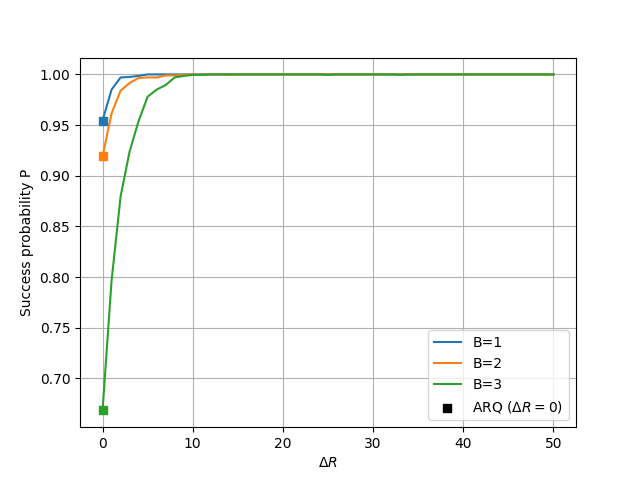}
\caption{Block success probability versus redundancy margin under different numbers of concurrent blocks.}
\label{fig:vary_B}
\end{figure}


Fig.~\ref{fig:vary_B} illustrates the impact of concurrently transmitted blocks. When multiple blocks share the same bandwidth, the effective symbol transmission time increases, thereby shrinking the repair budget available to each block.

Under this shared-bandwidth regime, ARQ becomes increasingly sensitive to budget reduction, and its success probability decreases noticeably as the number of blocks increases. DIR compensates for the reduced repair opportunities by increasing per-round innovation injection. A moderate redundancy margin restores near-unity reliability even when multiple blocks compete for bandwidth.

This behavior is consistent with the deadline-constrained repair model: reliability is governed jointly by per-round repair probability and the number of feasible rounds. Incremental redundancy provides robustness against budget shrinkage, which is common in dense V2X deployments.

\subsection{Ablation Study}

\begin{figure}[t]
    \centering
    \includegraphics[width=0.9\columnwidth]{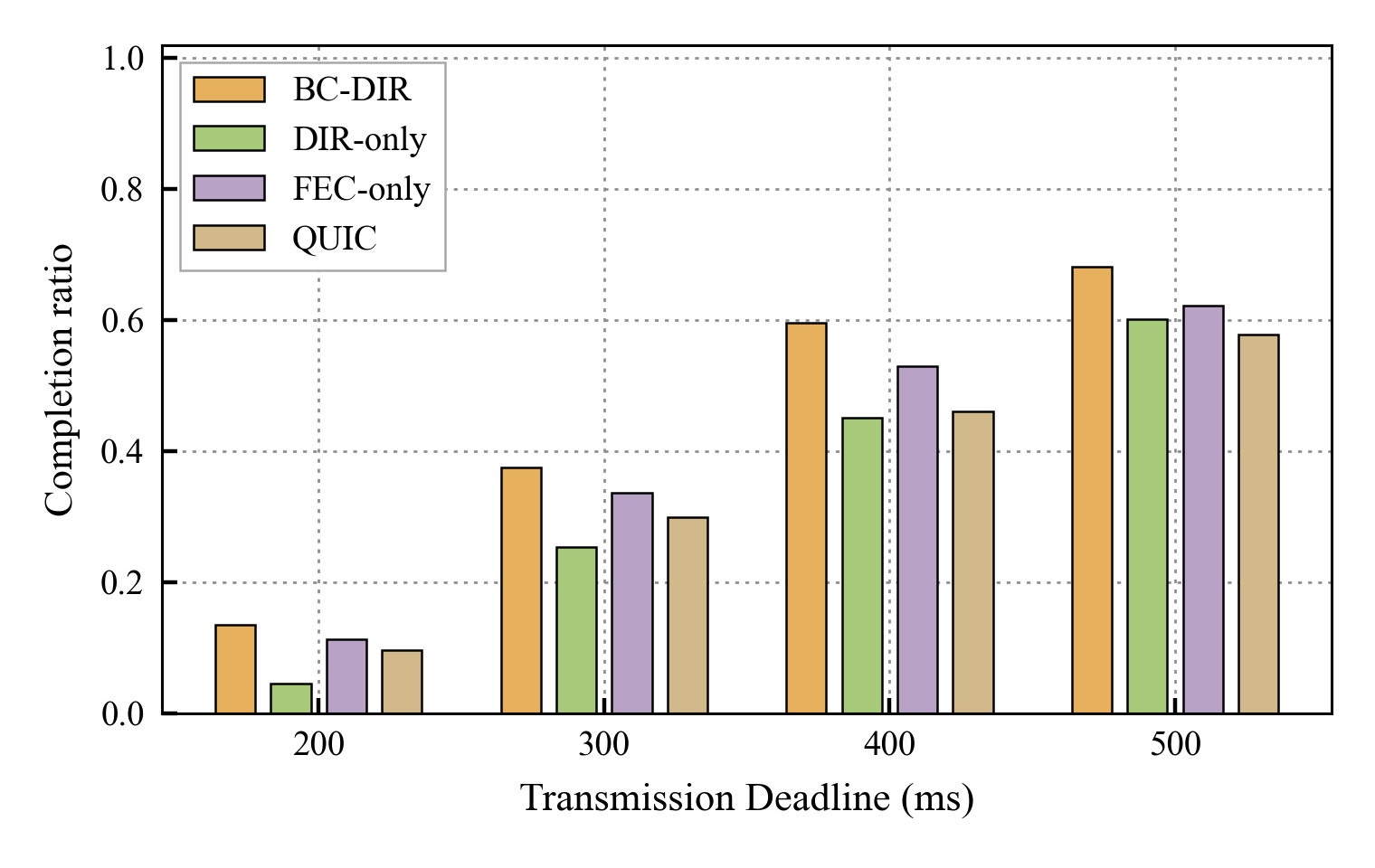}
    \caption{Completion ratio of QUIC, FEC-only, DIR-only, and BC-DIR under different transmission deadlines (traffic intensity = 600 vehicles/hour).}\label{fig:ablation_completion}
\end{figure}

\begin{figure}[t]
    \centering
    \includegraphics[width=0.9\columnwidth]{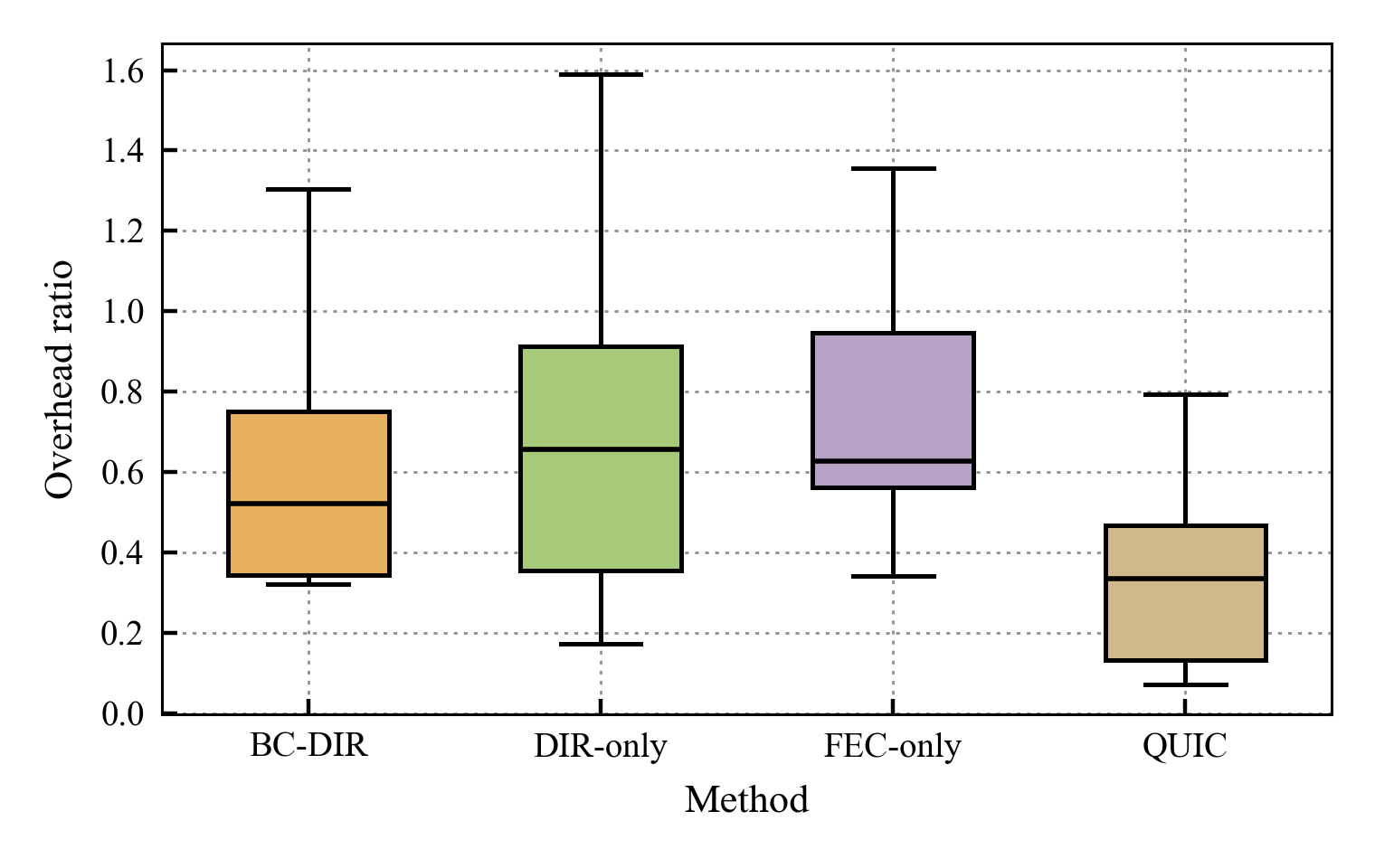}
    \caption{Transmission overhead of QUIC, FEC-only, DIR-only, and BC-DIR in the congested urban V2X scenario (traffic intensity = 600 vehicles/hour).}\label{fig:ablation_overhead}
\end{figure}

\begin{figure*}[!t]
\centering
\subfloat[][300 vehicles/hour]{\includegraphics[width=0.5\columnwidth]{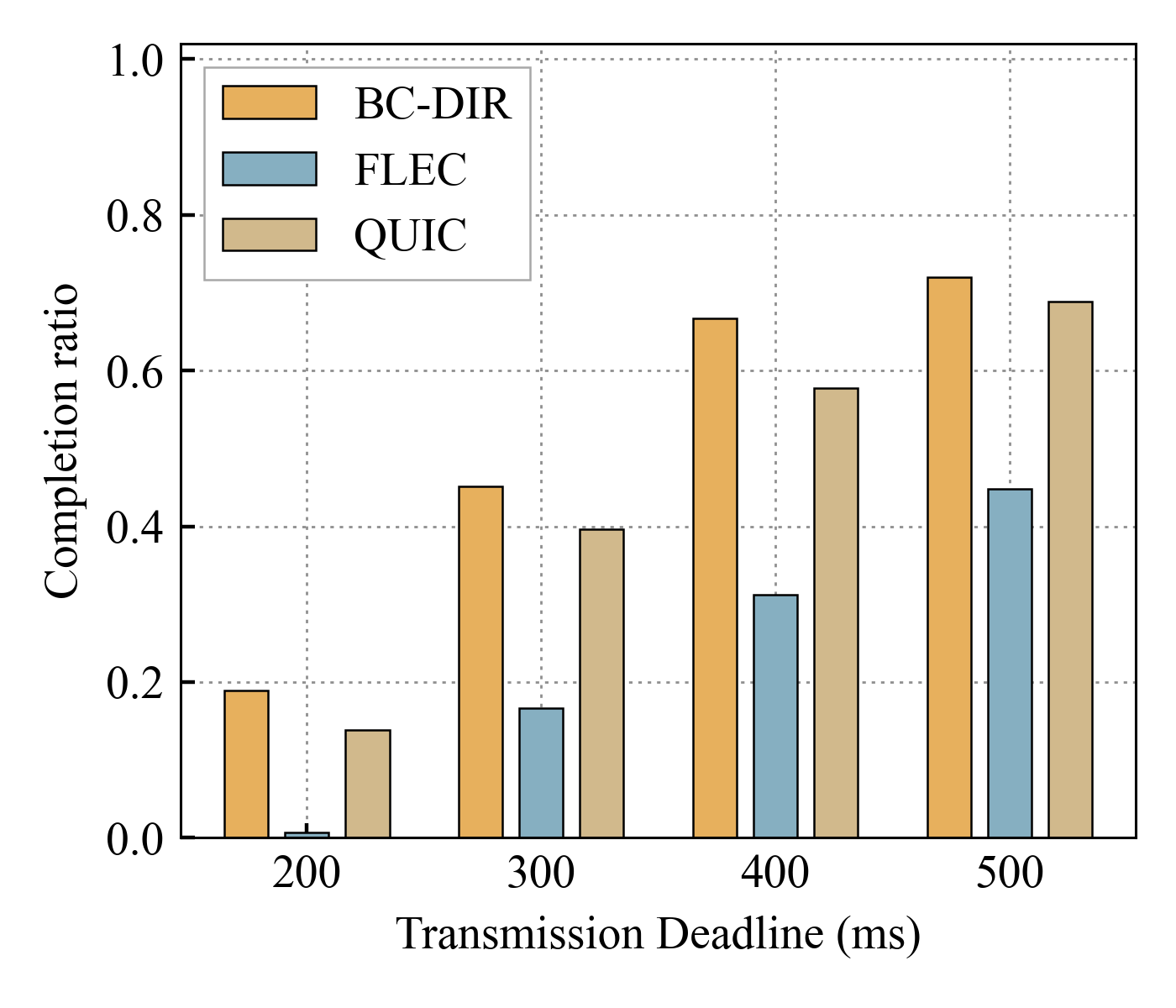}\label{fig:v2x-completion-d1}}
\hfill
\subfloat[][600 vehicles/hour]{\includegraphics[width=0.5\columnwidth]{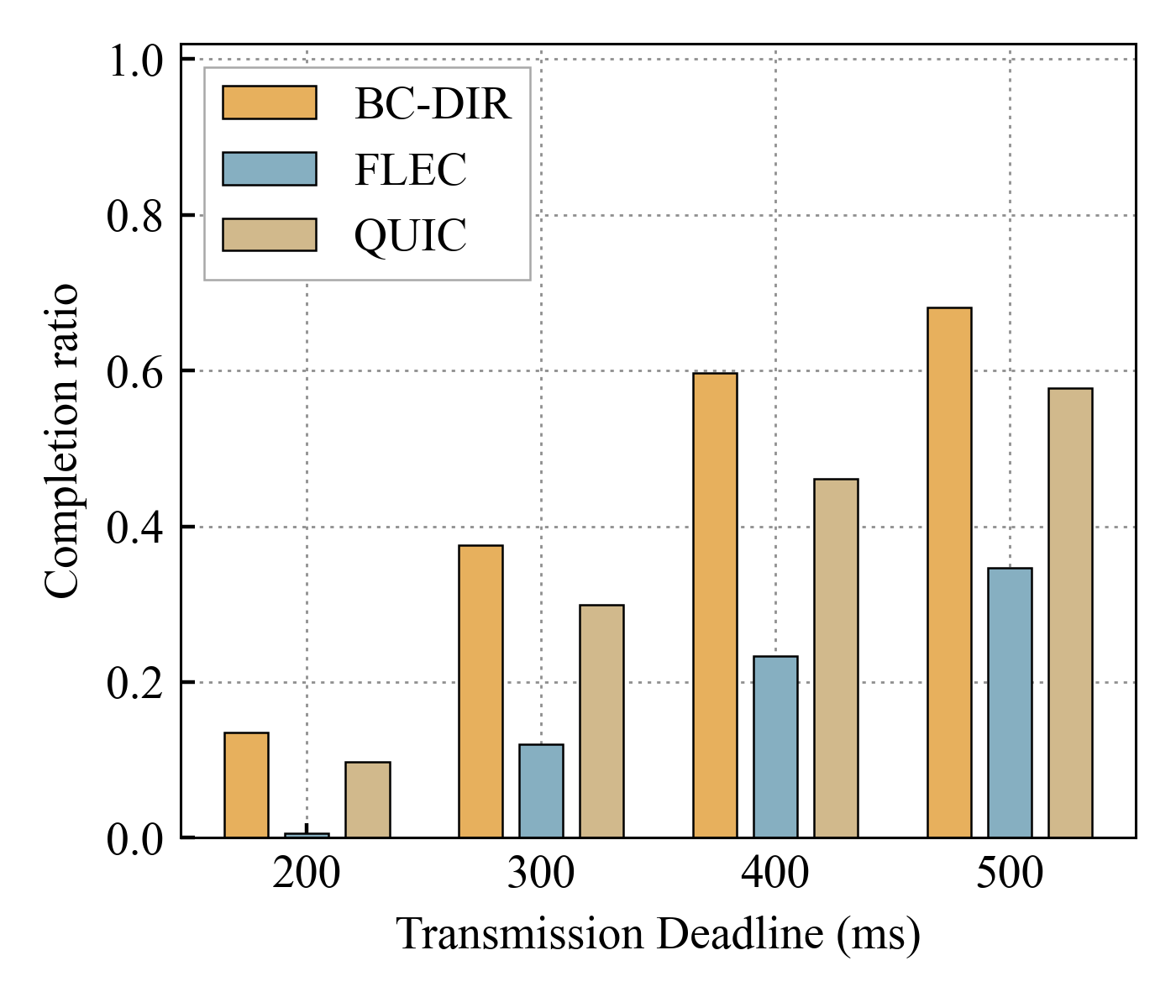}\label{fig:v2x-completion-d2}}
\hfill
\subfloat[][900 vehicles/hour]{\includegraphics[width=0.5\columnwidth]{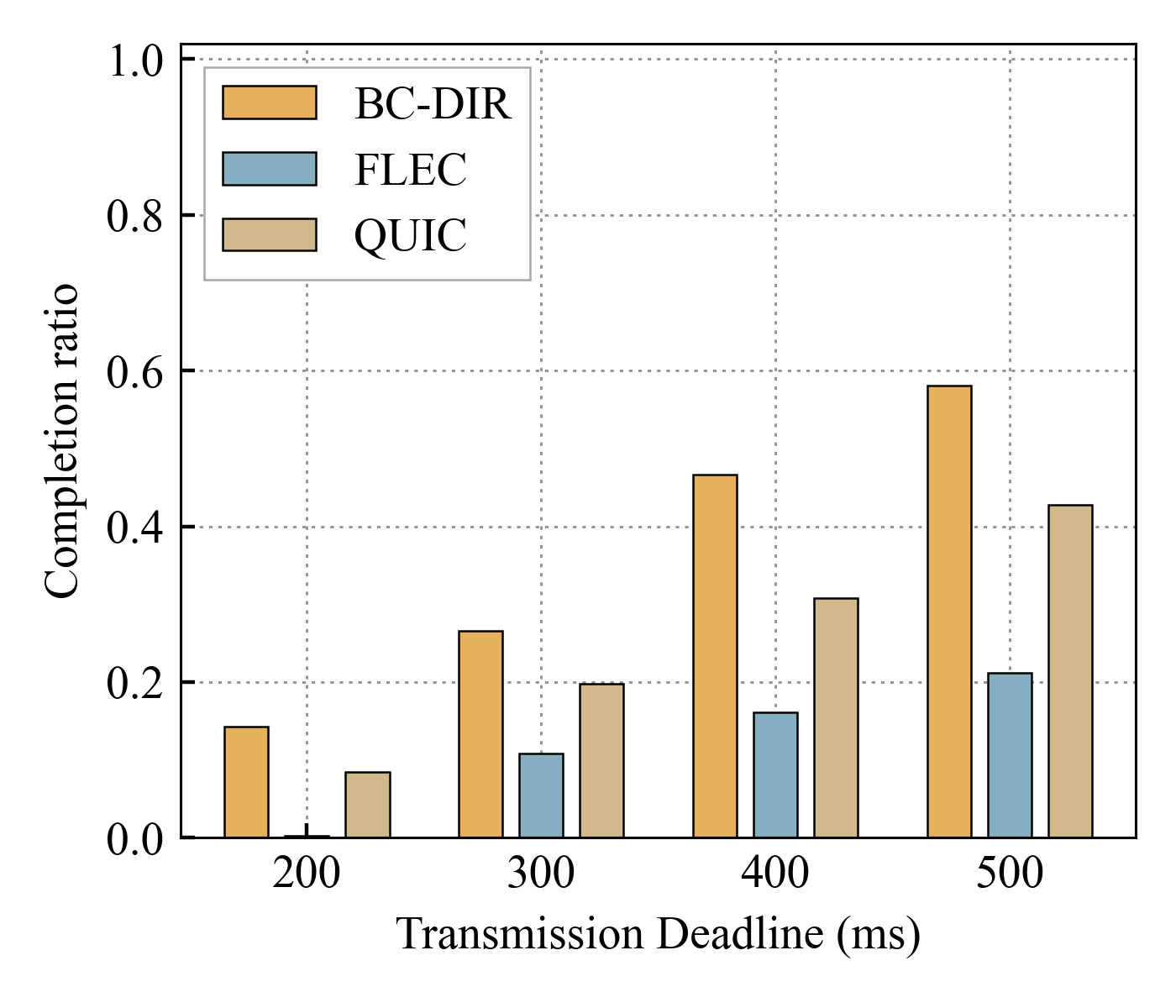}\label{fig:v2x-completion-d3}}
\hfill
\subfloat[][1200 vehicles/hour]{\includegraphics[width=0.5\columnwidth]{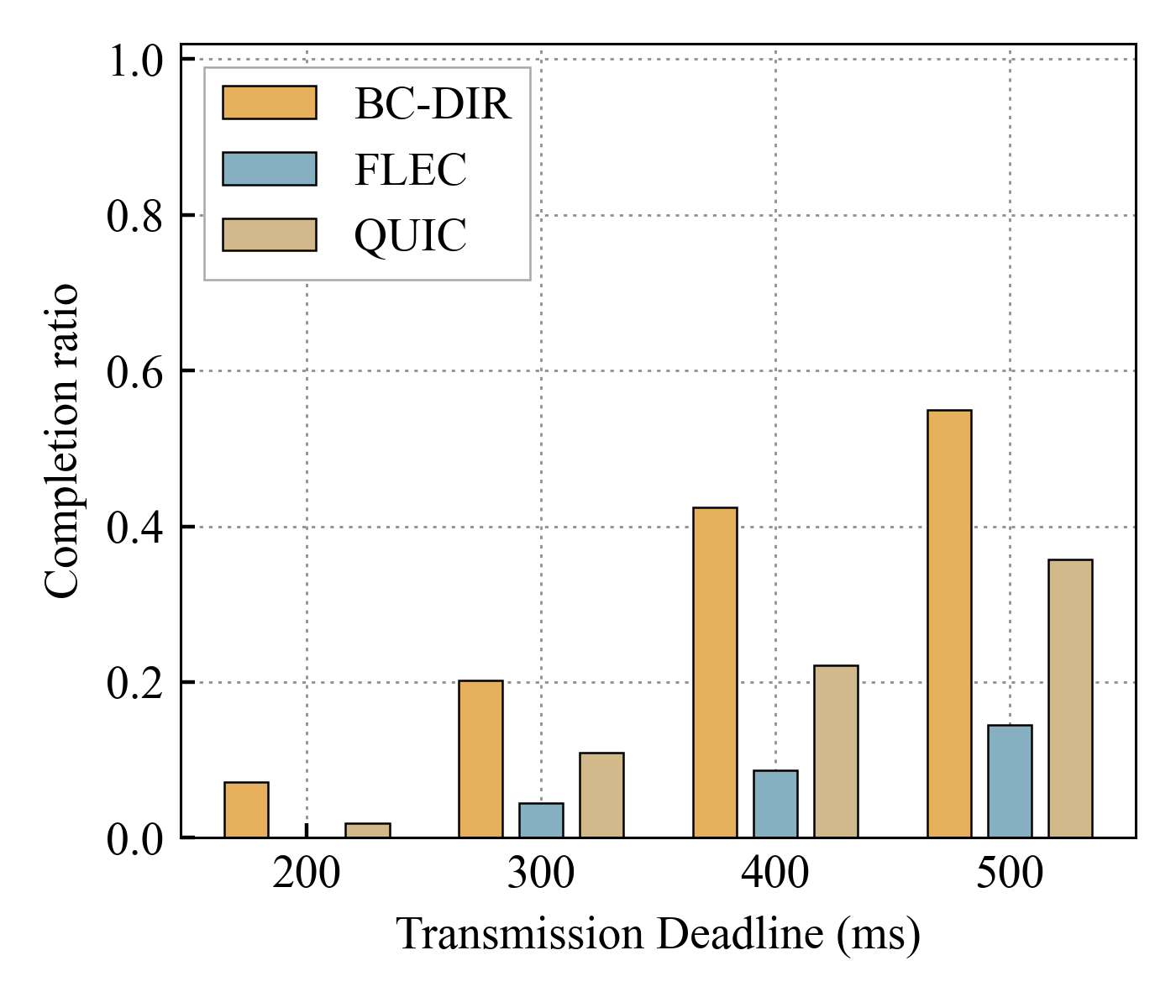}\label{fig:v2x-completion-d4}}
\hfill
\caption{Completion ratio under different transmission deadlines and traffic intensities in congested urban V2X scenarios for 100\,KB messages.}
\hfill
\label{fig:v2x-completion_group}
\end{figure*}

To further understand the contribution of each component in BC-DIR, we conduct an ablation study under the congested urban V2X scenario with traffic intensity fixed at 600 vehicles/hour.

Fig.~\ref{fig:ablation_completion} reports the completion ratio under different transmission deadlines. BC-DIR consistently achieves the highest completion ratio across all deadline settings, showing that the joint use of initial FEC protection and deadline-aware incremental repair is more effective than using either mechanism alone. Under tight deadlines, the gain is particularly clear. For example, at 200~ms, BC-DIR achieves about 0.14 completion ratio, compared with about 0.11 for FEC-only, 0.05 for DIR-only, and 0.10 for QUIC. As the deadline increases, all schemes improve, but BC-DIR maintains the best performance and reaches about 0.68 at 500~ms, outperforming DIR-only (0.60), FEC-only (0.62), and QUIC (0.58). These results indicate that initial FEC and subsequent DIR provide complementary benefits: proactive redundancy improves early decoding, while incremental repair increases the probability of successful recovery when burst loss exceeds the initial protection capability.

Fig.~\ref{fig:ablation_overhead} shows the corresponding overhead distribution. As expected, QUIC has the lowest overhead because it does not proactively inject coding redundancy. However, this low overhead comes at the cost of lower completion ratio. Among the FEC-enabled schemes, BC-DIR exhibits the lowest median overhead, clearly below both FEC-only and DIR-only. This result shows that the complete design can better regulate redundancy usage: instead of relying excessively on either proactive protection or repeated repair rounds, BC-DIR balances the two mechanisms and achieves the most favorable tradeoff between deadline-constrained completion ratio and transmission overhead.

\subsection{Performance in Congested Urban V2X Scenario}

\begin{figure*}[!t]
\centering
\subfloat[][300 vehicles/hour]{\includegraphics[width=0.5\columnwidth]{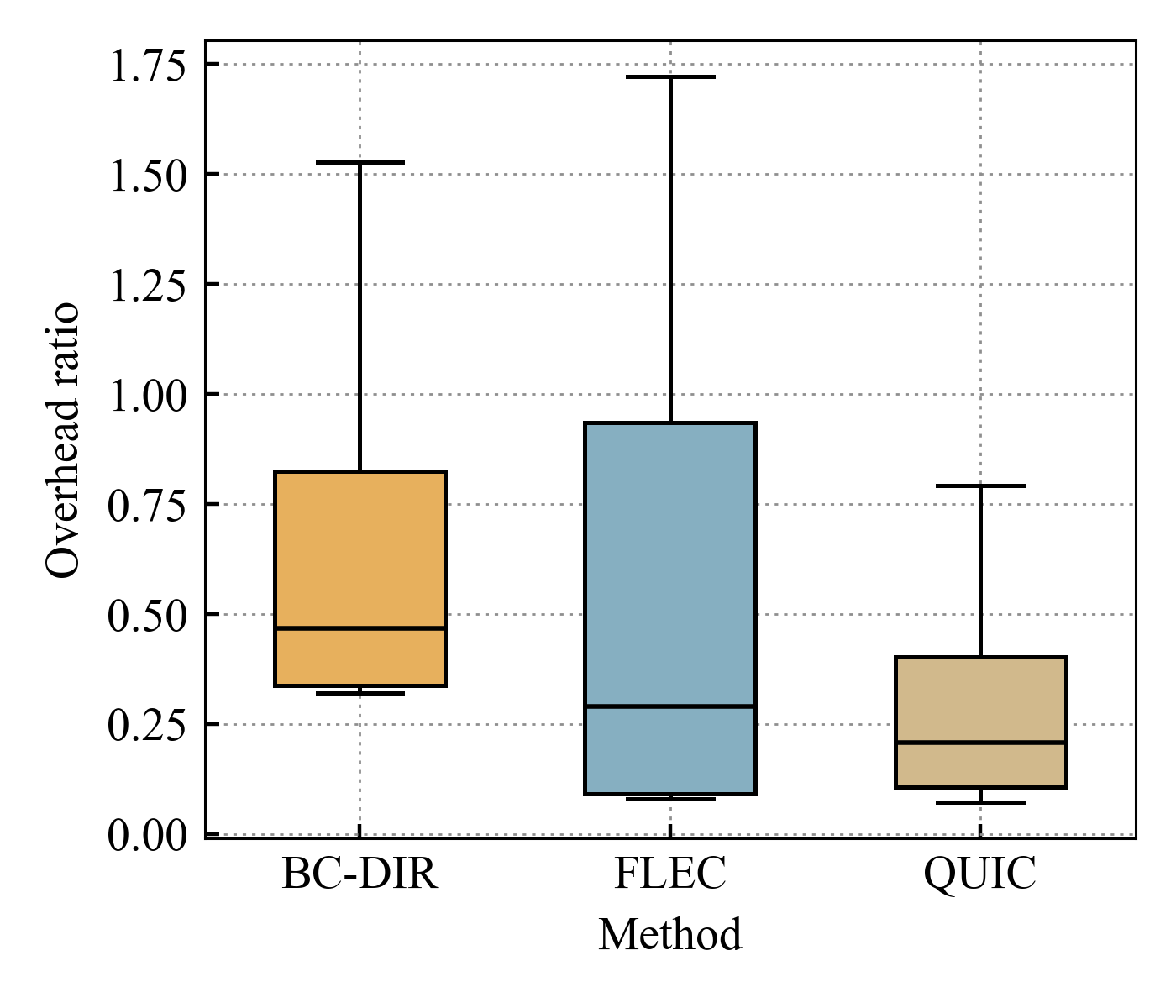}\label{fig:v2x-overhead-d1}}
\hfill
\subfloat[][600 vehicles/hour]{\includegraphics[width=0.5\columnwidth]{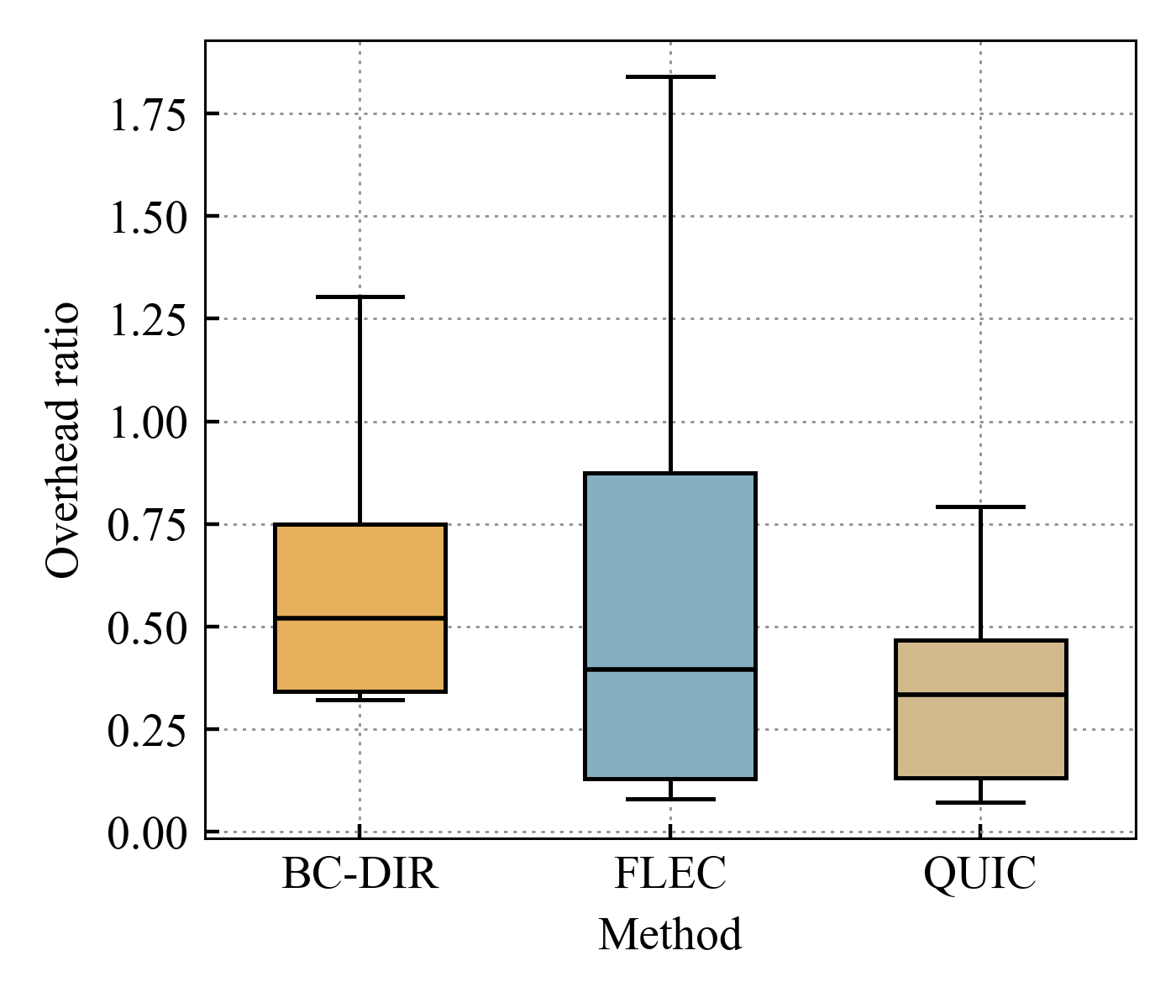}\label{fig:v2x-overhead-d2}}
\hfill
\subfloat[][900 vehicles/hour]{\includegraphics[width=0.5\columnwidth]{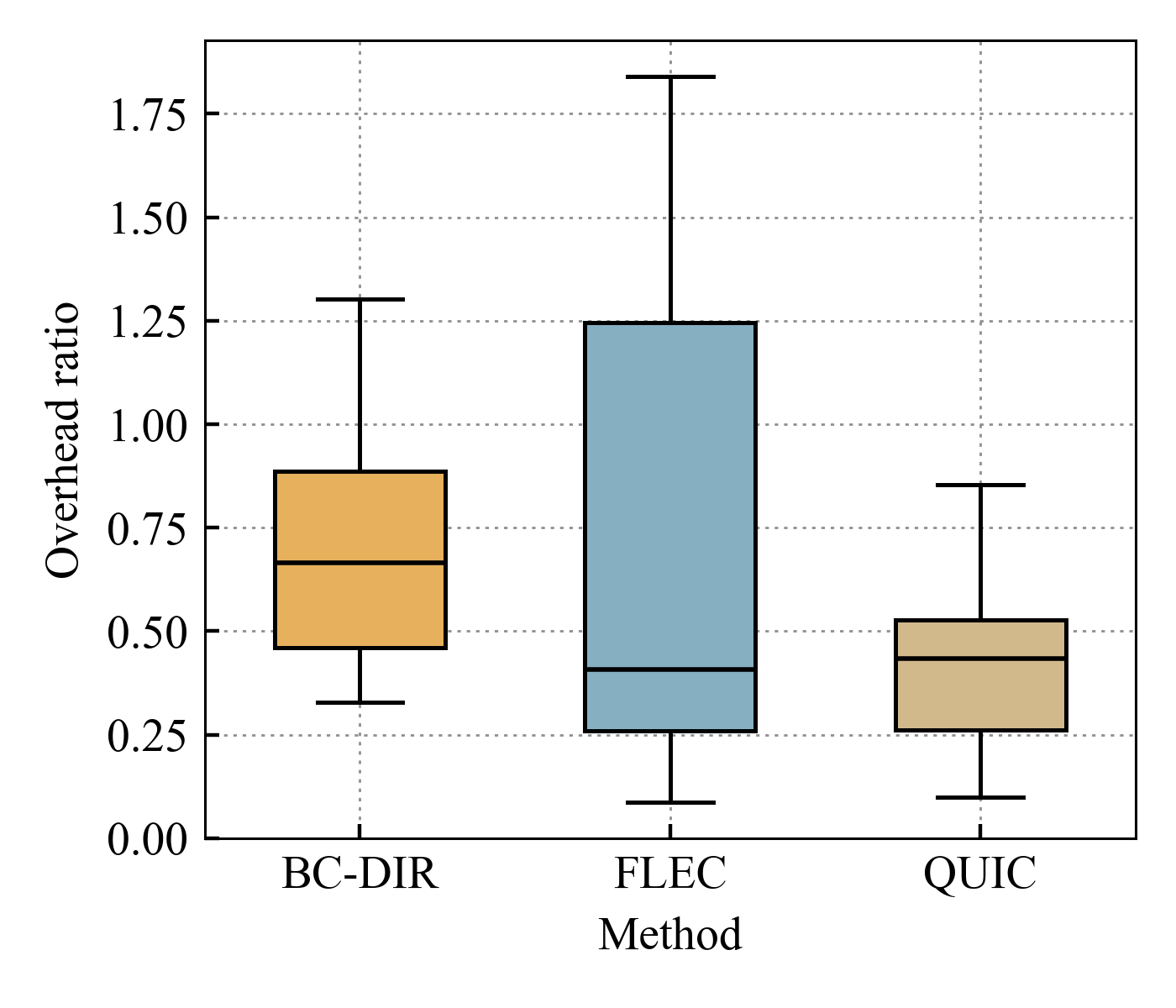}\label{fig:v2x-overhead-d3}}
\hfill
\subfloat[][1200 vehicles/hour]{\includegraphics[width=0.5\columnwidth]{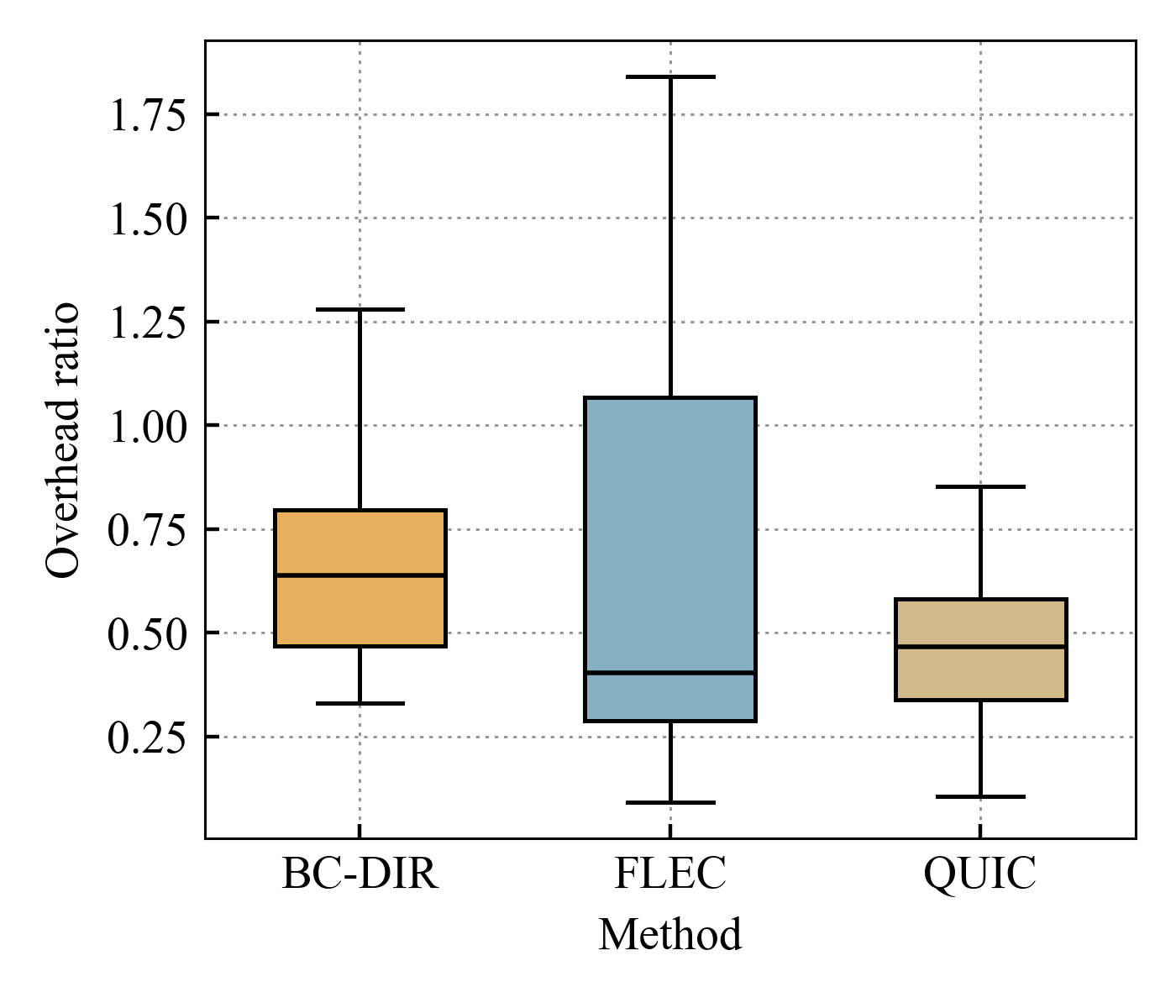}\label{fig:v2x-overhead-d4}}
\hfill
\caption{Transmission overhead ratio under different traffic intensities in congested urban V2X scenarios for 100\,KB messages.}
\label{fig:v2x-overhead_group}
\end{figure*}

We next evaluate BC-DIR under the congested urban V2X scenario, focusing on latency-sensitive medium-sized data transmissions of about 100\,KB, and compare BC-DIR with FLEC and native QUIC under different traffic intensities. For each setting, the simulation is independently repeated 30 times, and the reported results are obtained by aggregating the transmission statistics of all vehicles across these runs.

Fig.~\ref{fig:v2x-completion_group} presents the completion ratio under different transmission deadlines and traffic intensities. In V2X systems, medium-sized perception sharing or cooperative awareness updates typically require delivery within $200$--$500\,\mathrm{ms}$, and shorter deadlines imply stronger latency sensitivity. As expected, all schemes benefit from relaxed deadlines, but their performance diverges significantly as traffic intensity increases and the wireless channel becomes more bursty and congested.

Across all tested settings, BC-DIR consistently achieves the highest completion ratio. Under moderate traffic load (300 vehicles/hour), the gain over native QUIC is already visible, especially at 300--400\,ms deadlines. As traffic intensity increases to 600, 900, and 1200 vehicles/hour, the advantage of BC-DIR becomes increasingly pronounced. For example, at 500\,ms, BC-DIR achieves about 0.68, 0.58, and 0.55 completion ratio under 600, 900, and 1200 vehicles/hour, respectively, whereas FLEC remains below about 0.35, 0.22, and 0.15, and native QUIC reaches about 0.58, 0.43, and 0.36. The gap is even more significant under tighter deadlines, where retransmission delays become more damaging.

Fig.~\ref{fig:v2x-overhead_group} compares the transmission overhead ratio under different traffic intensities. As the traffic load increases from 300 to 1200 vehicles/hour, the overhead of all schemes tends to rise because higher contention and burstier losses require more recovery effort. Across all traffic intensities, QUIC consistently incurs the lowest overhead, since it does not proactively inject coding redundancy. However, this lower overhead comes at the cost of significantly reduced completion ratio under the same conditions.

Compared with QUIC, BC-DIR introduces a moderate overhead increase, with the median overhead remaining roughly 15\%--25\% higher across the tested traffic intensities. This additional cost is substantially lower and more stable than that of FLEC, whose overhead distribution is consistently broader and exhibits much larger upper tails, especially under 900 and 1200 vehicles/hour. In contrast, BC-DIR maintains a more concentrated overhead distribution, indicating that its redundancy adaptation is more stable under traffic-induced burst loss.

When considered together with the completion-ratio results in Fig.~\ref{fig:v2x-completion_group}, these results show that BC-DIR achieves a more favorable reliability-overhead trade-off than both benchmark methods. Relative to QUIC, the moderate increase in overhead yields a substantial improvement in deadline-constrained completion probability. Relative to FLEC, BC-DIR achieves higher completion ratio with a more controlled and stable redundancy cost, which is particularly important in highly congested V2X environments.

\begin{figure}[t]
\centering
\includegraphics[width=0.9\columnwidth]{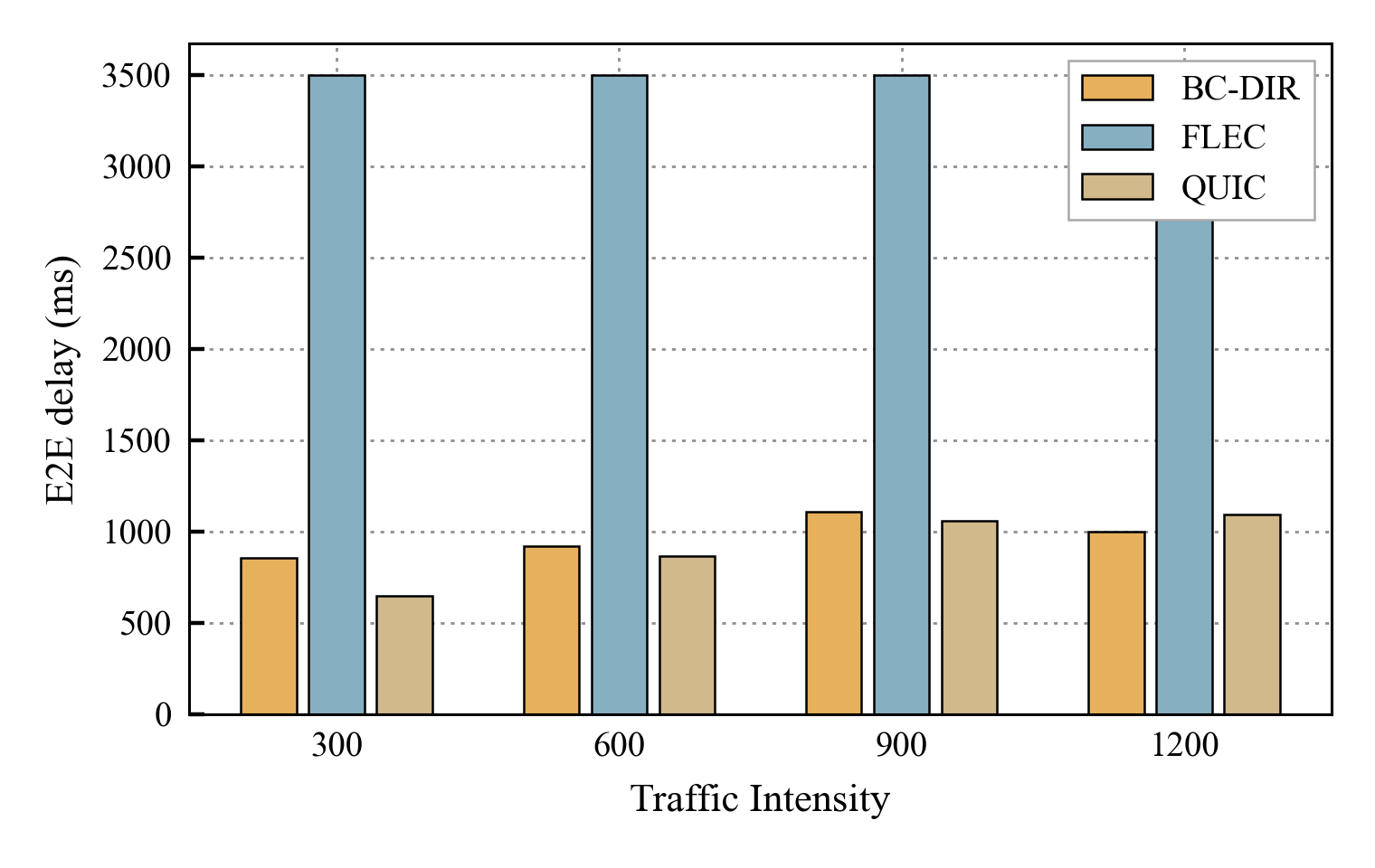}
\caption{95th-percentile end-to-end delay under different traffic intensities in congested urban V2X scenarios for 100\,KB messages.}
\label{fig:v2x_worst}
\end{figure}

To further examine worst-case performance, Fig.~\ref{fig:v2x_worst} shows the worst-case end-to-end delay under different traffic intensities. BC-DIR consistently maintains substantially lower tail delay than FLEC across all traffic loads. In particular, the worst-case delay of FLEC remains around 3500\,ms in all settings, indicating severe instability once burst losses and traffic contention accumulate. In contrast, BC-DIR keeps the worst-case delay below about 1.1\,s, even at 900 and 1200 vehicles/hour. Compared with native QUIC, BC-DIR achieves lower worst-case delay at 300 and 600 vehicles/hour, and remains competitive under heavier traffic loads. 

\begin{figure}[t]
\centering
\includegraphics[width=0.85\columnwidth]{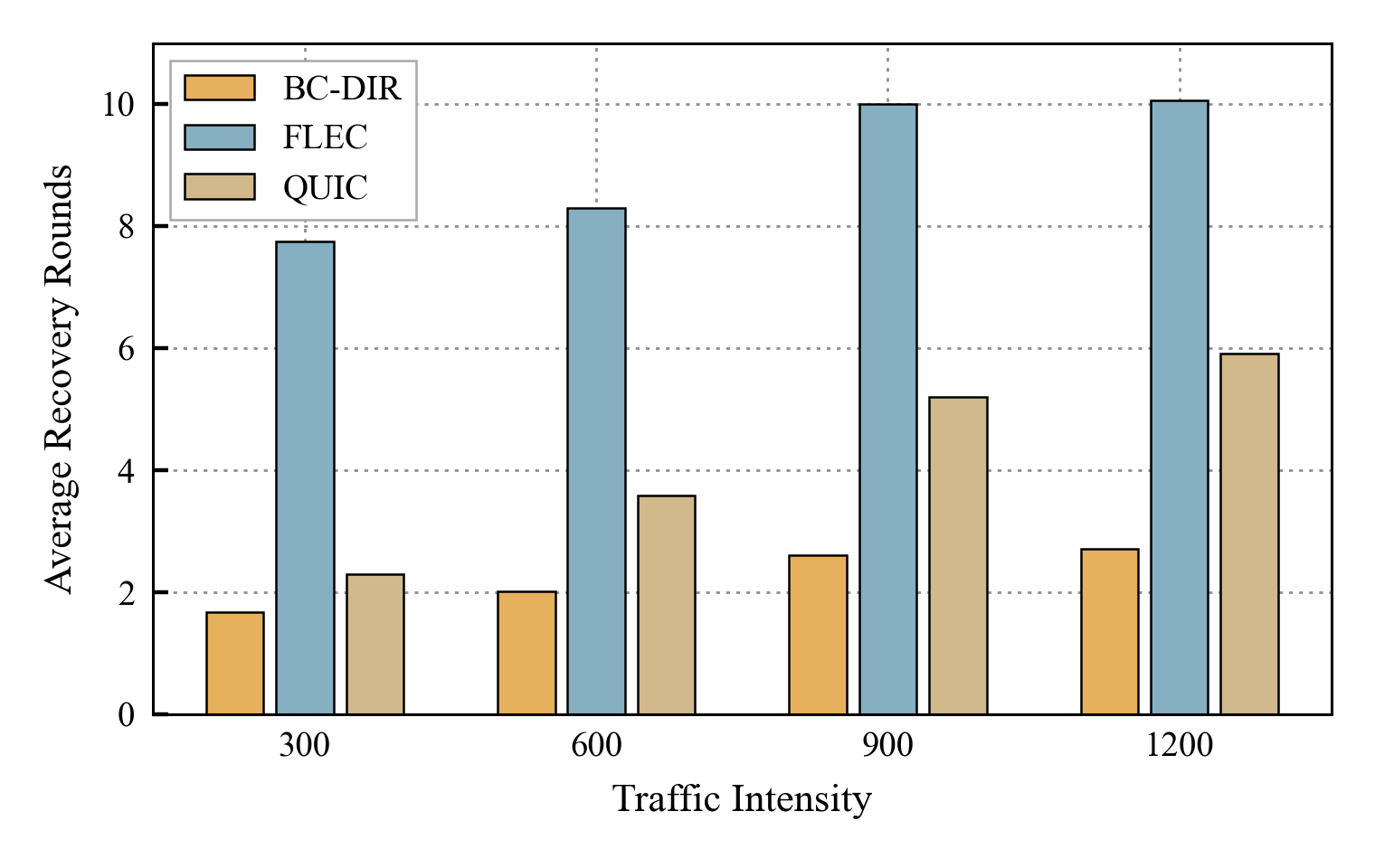}
\caption{Average recovery rounds under different traffic intensities in congested urban V2X scenarios for 100\,KB messages.}
\label{fig:v2x_retx}
\end{figure}

Fig.~\ref{fig:v2x_retx} reports the average number of recovery rounds under different traffic intensities. BC-DIR consistently requires the fewest recovery rounds across all traffic loads, and the advantage becomes more pronounced as traffic intensity increases. Compared with native QUIC, BC-DIR uses less than half as many recovery rounds on average, and only about one quarter as many as FLEC. Even at the highest traffic intensity of 1200 vehicles/hour, BC-DIR keeps the average number of recovery rounds below 3, whereas native QUIC rises to nearly 6 and FLEC exceeds 10. This result highlights the key benefit of the proposed repair mechanism: by injecting additional repair symbols in each round, BC-DIR increases the probability of successful recovery per round and avoids the repeated feedback cycles that dominate conventional retransmission-based recovery under burst loss.

\subsection{Performance Evaluation under Stable Network Conditions}

Even under predominantly line-of-sight and favorable communication conditions, small IID packet losses ($0.1\%\sim0.5\%$) are commonly observed in practical V2X systems. In this scenario, all schemes achieve $100\%$ completion ratio for both $100\,\mathrm{KB}$ and $1\,\mathrm{MB}$ transmissions.  As a result, reliability is no longer the primary differentiating factor. The key performance metric becomes transmission efficiency, which is reflected by redundancy overhead.

\begin{figure}[t]
\centering
\includegraphics[width=0.85\columnwidth]{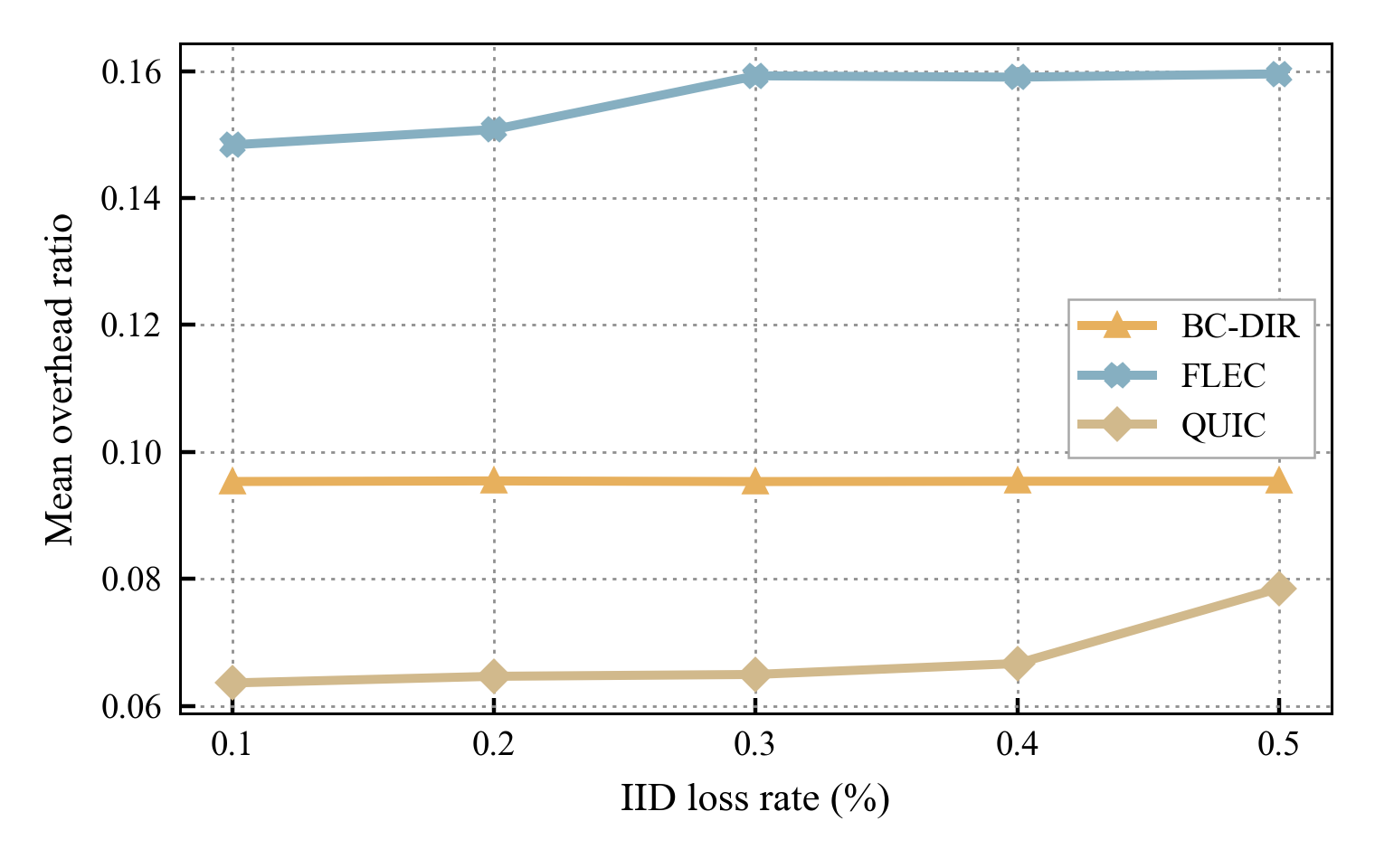}
\caption{Redundancy overhead for 100 KB objects under IID loss conditions.}
\label{fig:iid_100k_overhead}
\end{figure}

For $100\,\mathrm{KB}$ objects, the transfer typically completes within one to two RTTs, leaving limited opportunity for predictive redundancy to amortize across multiple repair rounds. As shown in Fig.~\ref{fig:iid_100k_overhead}, native QUIC incurs the lowest overhead, while BC-DIR increases the overhead by only about $3\%$ to provide robustness against residual random losses. In contrast, FLEC introduces a much larger overhead increase of approximately $8\%$ over QUIC across the tested i.i.d. loss range. This result indicates that for short transfers under stable channel conditions, the benefit of incremental redundancy is inherently limited, but BC-DIR is still able to maintain a substantially lower redundancy cost than existing FEC-based alternatives.

\begin{figure}[t]
\centering
\includegraphics[width=0.85\columnwidth]{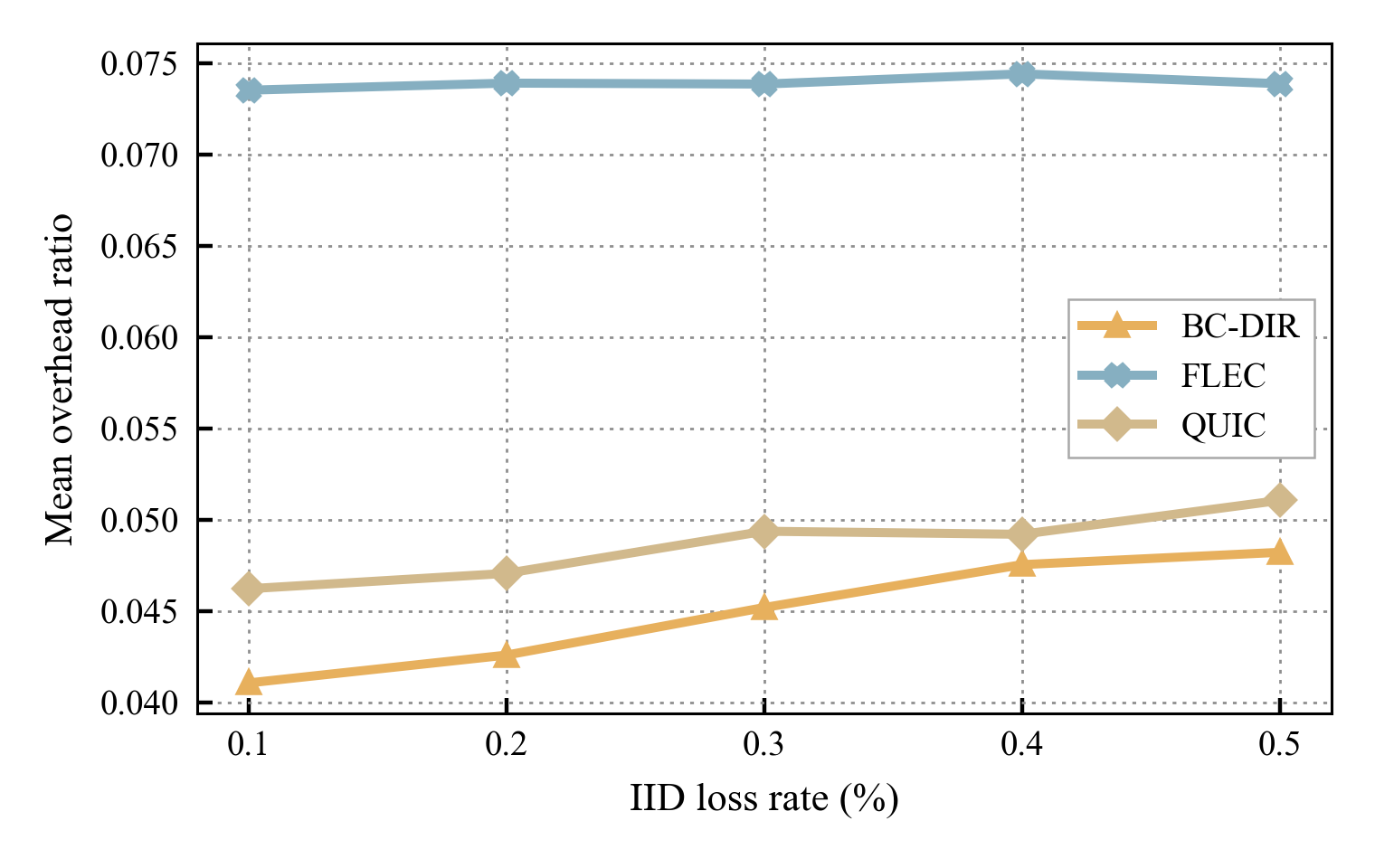}
\caption{Redundancy overhead for 1 MB objects under IID loss conditions.}
\label{fig:iid_1mb_overhead}
\end{figure}

For $1\,\mathrm{MB}$ objects, the transfer spans multiple RTT cycles, making cumulative recovery behavior more significant. In this case, BC-DIR leverages feedback across rounds to regulate redundancy injection more effectively. As shown in Fig.~\ref{fig:iid_1mb_overhead}, BC-DIR not only avoids introducing extra overhead relative to native QUIC, but in fact reduces the total overhead by about $1\%$. This gain comes from proactively transmitting an appropriate amount of repair symbols, which reduces repeated NACK-triggered retransmissions and the associated control overhead over long transfers. In contrast, FLEC maintains a consistently higher redundancy level than BC-DIR across the tested i.i.d. loss range.

Overall, these results show that under stable wireless conditions, the advantage of BC-DIR lies not in improving completion ratio, but in maintaining transmission efficiency. For short transfers, the benefit of adaptive redundancy remains limited because the transmission duration is too short to amortize coding overhead. For long transfers, however, BC-DIR achieves a more favorable redundancy-efficiency tradeoff by suppressing repeated repair signaling and regulating redundancy more precisely across rounds.
 
\section{Conclusion and Future Work}
\label{sec:conclusion}
This paper presented BC-DIR, a Bandit-Controlled Deadline-Aware Incremental Redundancy framework for QUIC-based V2X transport. BC-DIR extends conventional retransmission into a deadline-aware repair process by transmitting additional repair symbols after decoding failure, rather than retransmitting the original lost packets. By combining rateless coding, soft decoding deadlines, and online redundancy adaptation, the proposed design improves the probability that messages are completed within their delivery deadlines under burst-loss-dominated V2X conditions.

In addition to the protocol design, we developed a deadline-constrained reliability model that characterizes the benefit of incremental redundancy over conventional ARQ and establishes the existence of an optimal redundancy margin. The analytical results were further validated by Monte Carlo simulations under Gilbert--Elliott burst-loss channels. We also implemented BC-DIR in a QUIC-based transport stack and evaluated it in realistic vehicular simulations. Experimental results showed that BC-DIR consistently improves deadline-constrained completion probability in congested urban scenarios while maintaining competitive overhead under stable network conditions.

Several directions remain for future work. First, the current design focuses on single-path QUIC transport, while multi-path V2X communication may provide additional diversity for deadline-constrained recovery. Second, the present analytical model characterizes block-level recovery behavior; extending it toward a tighter message-level completion analysis would further strengthen the theoretical foundation. Finally, integrating richer cross-layer context, such as channel quality indicators or MAC-level congestion signals, may further improve the efficiency of online redundancy adaptation.



\ifCLASSOPTIONcaptionsoff
  \newpage
\fi

\bibliographystyle{IEEEtran}
\bibliography{bibtex/bib/IEEEabrv,bibliography}




\end{document}